\documentclass[journal]{IEEEtran}
\usepackage{amsmath,amssymb,amsfonts,amsthm}

\usepackage{bm}
\usepackage{graphicx}
\usepackage{booktabs}
\usepackage{xr-hyper}
\usepackage{algorithm}
\usepackage{algpseudocode}
\usepackage{hyperref}
\usepackage{float}
\usepackage{multirow}
\usepackage{stfloats}    
\usepackage{environ, wrapfig} 
\usepackage{bbm}

\usepackage[utf8]{inputenc}
\usepackage{xcolor}
\newtheorem{theorem}{Theorem}
\newtheorem{lemma}{Lemma}

\usepackage{refcount}
\usepackage{enumitem}
\newtheorem{assumption}{Assumption}

\providecommand{\abs}[1]{\lvert#1\rvert}
\renewcommand{\b}[1]{\ensuremath{\mathbf{#1}}} 
\newcommand{\bs}[1]{\ensuremath{\boldsymbol{#1}}} 
\renewcommand{\c}[1]{\ensuremath{\mathcal{#1}}} 
\newcommand{\Ex}[1]{\ensuremath{\mathbb{E}[#1]}}  
\newcommand{\Exp}[1]{\ensuremath{\mathbb{E}'[#1]}}  
\newcommand{\Px}[1]{\ensuremath{\mathbb{P}\left[#1\right]}}  
\newcommand{\ind}{1\hspace{-1.6mm}1} 
\newcommand{\norm}[1]{\ensuremath{\left\|#1\right\|}} 
\renewcommand{\O}[1]{\ensuremath{\mathcal{O}\left(#1\right)}} 
\newcommand{\Ot}[1]{\ensuremath{\tilde{\mathcal{O}}\left(#1\right)}} %
\newcommand{\leqtext}[1]{\ensuremath{\stackrel{\text{#1}}{\leq}}} 
\providecommand{\ip}[2]{\langle #1, #2 \rangle} 
\def \x {{\b{x}}}

\def \u {{\b{u}}}
\def \v {{\b{v}}}
\def \g {{\b{g}}}
\def \y {{\b{y}}}
\def \z {{\b{z}}}
\def \zero {{\b{0}}}
\def \one {{\b{1}}}
\def \T {{\mathsf{T}}}

\def \d {{\mathrm{d}}}
\def \p {{\mathrm{p}}}

\def \Rh {{\hat{R}}}

\def \cA {{\c{A}}}
\def \cB {{\c{B}}}

\def \cP {{\c{P}}}
\def \cM {{\c{M}}}

\def \cL {{\c{L}}}
\def \cN {{\c{N}}}

\def \cOt {{\tilde{\c{O}}}}

\def \cS {{\c{S}}}

\def \cR {{\c{R}}}

\def \cZ {{\c{Z}}}

\def \sX {{\mathsf{X}}}

\def \lam {{\bs{\lambda}}}
\def \lamh {{\hat{\lam}}}
\def \lamt {{\tilde{\lam}}}
\def \blamdt {{\bar{\lam}_\d^T}}
\def \blamt {{\bar{\lambda}^T}}
\def \phihk {{\hat{\phi}_K}}

\def \Lam {{\bs{\Lambda}}}

\def \EE {{\mathbb{E}}}
\def \Rn {{\mathbb{R}}}

\theoremstyle{definition}
\newtheorem{rem}{Remark}

\begin{document}

\title{Stochastic Geometry Meets Optimization: Data-Driven Cellular Planning}

\author{Abhishek~K.~Shukla, Abhishek~K.~Gupta,
and~Ketan~Rajawat%
\thanks{The authors are with the Department of Electrical Engineering,
Indian Institute of Technology Kanpur, Kanpur 208016, India
(e-mails: \texttt{abhisheks20@iitk.ac.in; gkrabhi@iitk.ac.in;
ketan@iitk.ac.in}). Part of this work was presented at the 2026 Allerton Conference. Supplementary material includes the MATLAB codes and additional proofs, available at \url{https://home.iitk.ac.in/~ketan/Supplementary}}%
}
\maketitle

\begin{abstract}
    Deployment optimization of base stations (BSs) is critical for improving cellular network performance metrics such as coverage and capacity, impacting user experience, scalability, and infrastructure cost-effectiveness. Traditional methods that optimize precise BS locations do not generally provide broader insights into spatial deployment priorities or alternative viable deployments. Stochastic geometry approaches, on the other hand, offer a more interpretable intensity-based framework, but rarely utilize real-world user location data or investigate optimal BS deployment strategies. In this work, we propose a data-driven stochastic geometry-based approach that optimizes BS deployment intensity to maximize a network performance metric using samples from an unknown user-location distribution. We specialize the framework to an inhomogeneous Poisson point process (PPP) for cellular coverage analysis. The resulting infinite-dimensional problem is discretized over the deployment area and solved using sampled user locations. We characterize the discretization and sampling errors, apply stochastic mirror ascent to solve the empirical problem, and analyze the selection of concrete BS deployments using independent validation samples. Unlike traditional methods, our approach yields an optimized spatial intensity that reveals deployment priorities across regions and accommodates practical deployment constraints, multi-tier networks, and incremental deployment scenarios. Numerical experiments with synthetic and real-world spatial demand data demonstrate the proposed approach under different network configurations.
\end{abstract}


\section{Introduction}
In cellular networks, deployment optimization refers to the problem of determining where to place the BSs or access points in a network, so as to maximize performance metrics such as user coverage and network capacity. As mobile networks become increasingly dense, heterogeneous, and complex, strategic BS placement plays a critical role in network planning, directly impacting user experience, scalability, and the cost-effectiveness of infrastructure investments \cite{lee2015base,Jurdi2017coverageholes}. The concept of deployment optimization encompasses a range of scenarios, including selectively deactivating underutilized BSs for energy savings \cite{suarez2015energy, le2017cell} and temporarily deploying mobile or UAV-based BSs to address localized surges in demand \cite{wang2019deployment, lyu2016placement, zhao2018deployment}, underscoring its broad applicability in modern network designs. 

Classical approaches have focused on finding the exact locations of BSs by directly optimizing a relevant performance metric, using a range of mathematical programming approaches \cite{huang2022deployment}, heuristics \cite{akram2020multicriteria}, and game-theoretic formulations \cite{fu2023joint}. These methods seek a single deployment for a given network state and feasible-site set. Such a solution is appropriate when the exact sites are known a priori, but provides limited guidance when site availability, deployment constraints, or network conditions change as the network evolves. In particular, it does not characterize alternative viable deployments or indicate which regions remain attractive for deployment under related network configurations.



Recently, stochastic geometry has emerged as a tractable approach to model and analyze wireless networks by accounting for various practical effects including BS locations, fading coefficients, shadowing, blockage, and antenna directionality, under reasonable assumptions \cite{andrews2023SGBook}. In particular, homogeneous Poisson point process models for users and BSs often yield simple expressions for network performance metrics such as coverage probability and rate coverage, and are widely used for analysis \cite{Dhillon2012hetnet, Gupta2016mmwave} as well as BS intensity optimization \cite{di2018system}. More recent stochastic-geometry models also account for the spatial coupling between user hotspots and BS locations arising in user-centric small-cell deployments \cite{saha2019unified}. However, these works primarily analyze prescribed spatial models rather than optimize the BS intensity using observed user locations. 

A key challenge in computing the optimal BS deployment is the lack of knowledge of the underlying mobile user distribution. In theoretical studies, it is common to make simplifying assumptions that scatter users according to a stationary uniform distribution \cite{andrews2023SGBook}. Such simplistic models, however, do not capture the behavior of real-world users, which may congregate at certain spots, avoid specific zones, or exhibit multi-scale behavior, motivating the need for data-driven modeling. In related contexts, recent studies have increasingly utilized actual user location data, collected from cellular traces \cite{hammami2016network} and Wi-Fi logs \cite{kulacz2022dynamic} over time, in order to build realistic user density maps using machine learning approaches such as kernel density estimation. We remark that an additional challenge here is that at any given time, only a subset of users is generally active and observable. 

In this work, we put forth a novel data-driven stochastic geometry-based approach for determining an optimal BS deployment to provide connectivity to users spread over a specified region. Despite the huge potential and multiple applications in 5G/6G network planning, such a formulation has never been considered in the literature. In particular, we seek to design an optimal BS deployment that provides coverage to users whose underlying location distribution is unknown but observed via a few sample locations. To this end, we model the BSs as following an  inhomogeneous PPP with intensity function $\lambda(\x)$. Our objective is to compute the optimal $\lambda(\x)$ that maximizes a specific network performance metric $R(\lambda)$, such as coverage or throughput. 

Unlike traditional deployment approaches that yield precise BS locations, our approach yields an intensity function that reflects the optimal BS intensity across regions, offering insights into the deployment needs (e.g. densification vs. interference management), while also adapting to actual user patterns. The intensity function can also be interpreted as a probabilistic priority function for BS soft-deployment to guide an operator. Given the BS intensities, exact BS locations may subsequently be extracted while taking into account other geographical, economic, and social constraints. The proposed approach is also amenable to varied practical deployment constraints, multi-tier networks, and incremental deployment scenarios. Remarkably, it is known that under sufficiently strong shadowing, even deterministic BS deployments may exhibit propagation characteristics similar to those of Poisson networks \cite{Bartek2015PPP}, providing further motivation for studying intensity-based soft deployment.

At the outset, determining the optimal BS intensity using real-world user data is challenging on two fronts: the intensity function $\lambda(\x)$ is infinite-dimensional and the underlying user distribution is unknown. To address the first challenge, we discretize the area of interest into $D$ smaller sub-regions, so that $\lambda(\x)$ can be approximated with a piecewise constant function, rendering the optimization finite-dimensional. Second, to bypass the need for the unknown user distribution, we approximate $R(\lambda)$ with the empirical average over the $M$ sampled user locations. While the solution of the approximate problem is suboptimal with respect to the original problem, we establish that the expected optimality gap is no more than $\epsilon$ provided that $D = \cOt(\epsilon^{-2})$ and $M = \cOt(\epsilon^{-4})$. Interestingly, these bounds are generally applicable to a broader class of variational optimization problems, and may be of independent interest. 

Next, we provide a stochastic mirror ascent algorithm for solving the finite-dimensional empirical problem efficiently. We then derive the coverage probability expressions for an arbitrary BS intensity under instantaneous SINR-based and average SINR-based cell association rules, which can serve as the objectives of our maximization. We also propose a procedure for generating concrete BS deployments from the optimized intensity and selecting a suitable realization using independent validation samples, while quantifying the associated placement error. We highlight various practical applications of the proposed approach in 5G/6G networks, including multi-tier deployments, incremental deployment in the presence of existing BSs, and deployment with exclusion zones. Numerical investigations demonstrate the usefulness of the approach and provide insights into how the optimized BS intensity depends on user distribution, existing infrastructure, and deployment constraints. In particular, the optimized intensity reveals not only preferred deployment regions but also alternative viable locations and their relative priorities, which are not evident from a single deterministic deployment. To the best of our knowledge, this is the first framework that learns a fully spatially varying BS point-process intensity directly from user-location samples while jointly quantifying discretization, statistical estimation, and realization-selection errors.


\noindent{\emph{Notation:}} We denote scalars by regular font letters, vectors by bold font letters, and sets by calligraphic capital letters. The notation $\norm{\cdot}$
is used for both $L^2$ and $\ell^2$ norms, depending on whether the argument is a function or a vector, respectively. 

\section{Problem Formulation} \label{system_model}
We consider the design of a cellular network where the BSs follow a point process $\Phi_\lambda$ parameterized by an intensity  function $\lambda:\cA\rightarrow \Rn_+$ where $\cA \subset \Rn^2$ is a compact region with area $A = \abs{\cA}$ and diameter $A_d$. The users are also distributed according to an independent point process $\Psi_u$ over $\cA$ with user locations following distribution $\z \sim P_\z$. Let $r(\lambda,\z)$ denote the performance metric of interest for a user located at $\z$, such as coverage probability (CP) $p_c = \mathbb{P}[\mathrm{SINR}>\tau]$ or per-user rate coverage $r_c = \mathbb{P}[\mathrm{Rate} > \tau]$. The average user performance is given by $R(\lambda)=\EE_{\z\sim P_\z}[r(\lambda,\z)]$. We seek an intensity $\lambda$ that maximizes $R(\lambda)$. 
Examples of $R$ include the average SINR (or rate) coverage which denotes the fraction of users that are in coverage, i.e., their SINR (or rate) is above the target threshold. The optimized intensity $\lambda$ provides a spatially adaptive guideline for BS deployment, where higher intensity values indicate regions requiring denser BS placement while lower values suggest sparser deployment.

The problem of optimally deploying an average of $N$ BSs in the region $\cA$ can be written as  
\begin{align}
	\lambda^\star &= \arg\max_{\lambda \in L_+^1(\cA)} \Ex{r(\lambda, \z)} \tag{$\cP$}\label{p}\\
	&\text{s. t. } \hspace{5mm}\int_{\cA} \lambda(\x) d\x  = N \label{pcons}
\end{align}
where $L_+^1(\cA)$ denotes the set of nonnegative Lebesgue-integrable functions on $\cA$. Denoting the set of admissible intensity functions by $\cS := \{\lambda \in L_+^1(\cA) \mid \int_{\cA} \lambda(\x) d\x  = N\}$, we can write \eqref{p} compactly as $\lambda^\star = \arg\max_{\lambda \in \cS} R(\lambda)$. To avoid trivial cases, we will assume that $\max_{\lambda \in \cS}R(\lambda)$ is finite and achieved by some $\lambda^\star \in \cS$.

As discussed earlier, without an explicit user distribution, we cannot directly solve \eqref{p}, as the expectation in the objective function cannot be evaluated. Instead, given $M$ independent and identically distributed (i.i.d.) user-location samples $\cZ := \{\z_m\}_{m=1}^M$, where $\z_m \sim P_\z$, we can replace $R(\lambda)$ by its empirical approximation $\hat{R}(\lambda) = \frac{1}{M}\sum_{m=1}^M r(\lambda,\z_m) \approx \Ex{r(\lambda,\z)}$. Such an empirical estimate accommodates arbitrary and spatially heterogeneous user distributions, and improves in accuracy as $M$ increases. To obtain a finite-dimensional problem suitable for numerical optimization, we also need to discretize the intensity function. To this end, we partition the region $\cA$ into $D$ disjoint equal-area sub-regions denoted by $\cA_1$, $\ldots$, $\cA_D$ such that $\cup_k \cA_k = \cA$. The piecewise constant approximation of $\lambda(\x)$ is denoted by $\lambda_\d(\x) := \sum_{k=1}^D \lambda_{\d k} \ind_{\x \in \cA_k}$  where 
\begin{align}\label{discretization}
	\lambda_{\d k} = \frac{1}{\abs{\cA_k}} \int_{\cA_k} \lambda(\x) d\x
\end{align}
and $\ind$ denotes the indicator function which takes the value 1 if the condition in its subscript is true and zero otherwise. Observe further that since $\abs{\cA_k} = A/D$, the left-hand side of the constraint in \eqref{pcons} simplifies as
\begin{align}
	\int_{\cA}\lambda(\x) d\x = \sum_{k=1}^D \int_{\cA_k} \lambda(\x)d\x = \frac{A}{D} \sum_{k=1}^D \lambda_{\d k}.
\end{align}
Therefore, the constraint \eqref{pcons} can be written as $\sum_k \lambda_{\d k} = ND/A$. Collecting $\lambda_{\d 1}$, $\ldots$, $\lambda_{\d D}$ into the vector $\lam_\d \in \Rn^D$, we see that $\lambda_\d \in \cS$ is equivalent to $\lam_\d \in \Lam := \{\lam_\d\succeq \zero \mid \one^\T\lam_\d= \frac{ND}{A}\}$. Hence, the modified problem becomes
\begin{align}
	&\lamh_\d = \arg\max_{\lam_\d \in \Lam} \frac{1}{M} \sum_{m=1}^M r(\lam_\d, \z_m) \label{pe}\tag{$\hat{\cP}$}
\end{align}
where, with some abuse of notation, we allow functionals $r$, $R$, and $\Rh$ to also take the vectors such as $\lam_\d$ and $\lamh_\d$ (in addition to the \emph{functions} $\lambda_\d$ and $\hat{\lambda}_\d$) as inputs, i.e., $r(\lam_\d, \z) := r(\lambda_\d, \z)$ and $R(\lam_\d ) := R(\lambda_\d)$. In a similar vein, the optimal value of \eqref{pe} is given by $\Rh(\lamh_\d) := \Rh(\hat{\lambda}_\d)$ where we define the function $\hat{\lambda}_\d(\x): = \sum_{k=1}^D \hat{\lambda}_{\d k} \ind_{\x \in \cA_k}$. Hence, \eqref{pe} can be seen as a discretized and sampled version of \eqref{p}. As \eqref{pe} is a finite-dimensional optimization problem, it can be solved using standard numerical methods. In the next section, we characterize the excess error associated with solving \eqref{pe} instead of \eqref{p}. In practice, the optimization algorithm itself may incur additional error, e.g., if \eqref{pe} cannot be solved exactly due to nonconvexity. We do not analyze that component of the error in this section. 

\section{Bound on the excess error}\label{sec-excess}
We define the excess error as $E = \Ex{\abs{R(\lambda^\star) - R(\hat{\lambda}_\d)}} = \Ex{\abs{R(\lambda^\star) - R(\lamh_\d )}}$ where the expectation is with respect to the random samples $\{\z_m\}_{m=1}^M$. As we describe below, this error consists of two factors arising from discretizing the intensity function and sampling the user distribution. This section characterizes $E$ in terms of the average BS count $N$, discretization level $D$, and number of user-location samples $M$, thereby quantifying how these parameters affect the quality of the learned BS intensity. 

To separate the effects of discretization and sampling, we  introduce the oracle discretized solution 
\begin{align}\label{lamd}
	\lamt_\d  = \arg\max_{\lam_\d\in \Lam} R(\lam_\d ),
\end{align}
which optimizes the exact average metric $R(\lam_\d)=\Ex{r(\lam_\d,\z)}$ over the same discretized feasible set $\Lam$, while the expectation is with respect to the true user distribution. Let $\tilde{\lambda}_\d$ denote the piecewise constant intensity function corresponding to $\lamt_\d$. Then,
\begin{align}
	E &= 	\Ex{\abs{R(\lambda^\star) - R(\lamt_\d) + R(\lamt_\d ) - R(\lamh_\d )}} \\
	&\leq \abs{R(\lambda^\star) - R(\tilde{\lambda}_\d)} + \Ex{\abs{R(\lamt_\d ) - R(\lamh_\d )}} \label{decom}
\end{align}
where we have used the triangle inequality and the fact that $R(\lamt_\d ) = R(\tilde{\lambda}_\d)$ as explained earlier. Here, the first term captures the loss due to restricting the original infinite-dimensional problem to piecewise constant intensities, while the second term captures the additional loss due to replacing the exact objective by its empirical approximation. Accordingly, we define the discretization and sampling errors as
\begin{align}
	E_\d &:= \abs{R(\lambda^\star) - R(\tilde{\lambda}_\d)}, \label{ed} \\
	E_s &:= \Ex{\abs{R(\lamt_\d ) - R(\lamh_\d )}}. \label{es}
\end{align}
The subsequent subsections provide bounds on these two error terms, respectively. 

\subsection{Bound on the discretization error} 
Let $\lam_\d^\star$ denote the vector obtained by discretizing the optimal $\lambda^\star$ as described earlier in \eqref{discretization} so that $\lambda_\d^\star(\x) = \sum_{k} \lambda_{\d k}^\star \ind_{\x \in \cA_k}$. Similarly, let $\tilde{\lambda}_\d(\x) := \sum_k \tilde{\lambda}_{\d k}\ind_{\x \in \cA_k}$. Since $\tilde{\lambda}_\d \in \cS$, it follows from the definition of $\lambda^\star$ that $R(\lambda^\star) \geq R(\tilde{\lambda}_\d)$. Likewise, since $\lam_\d^\star \in \Lam$, it follows from the definition of $\lamt_\d $ that $R(\tilde{\lambda}_\d) \geq R(\lambda_\d^\star)$. Combining these inequalities, we obtain
\begin{align}
	0 \leq R(\lambda^\star) - R(\tilde{\lambda}_\d) \leq R(\lambda^\star) - R(\lambda_\d^\star) \\
	\Rightarrow E_\d =\abs{R(\lambda^\star) - R(\tilde{\lambda}_\d)} \leq \abs{R(\lambda^\star) - R(\lambda_\d^\star)}. \label{edbound1}
\end{align}
It remains to bound $\abs{R(\lambda^\star)-R(\lambda_\d^\star)}$, which captures the loss incurred by replacing $\lambda^\star$ with its piecewise constant approximation. 

To obtain such a bound, we need to impose continuity of $R$ with respect to an appropriate metric. Since $\lambda$ describes the spatial distribution of a fixed BS mass over $\cA$, we can use the Wasserstein-1 distance metric, also known as the earth mover's distance. Specifically, let $\cL_1$ be the set of all 1-Lipschitz functions on $\cA$ so that 
\begin{align*}
	\cL_1 = \{\varphi:\cA \rightarrow \Rn\mid \abs{\varphi(\x) - \varphi(\y)} \leq \norm{\x-\y}, \ \forall \x,\y \in \cA\}
\end{align*}
where $\norm{\cdot}$ is the Euclidean norm. Then, the Wasserstein-1 distance between two intensities $\lambda$ and $\lambda'$ is given by 
\begin{align}
	W_1(\lambda,\lambda') = \sup_{\varphi \in \cL_1} \left\{\int_{\cA} \varphi(\x)(\lambda(\x)-\lambda'(\x)) d\x\right\}\label{Wdef}
\end{align}
where $\int_\cA\lambda(\x)\d\x = \int_{\cA}\lambda'(\x)\d\x = N$. Unlike the $L^2$ metric, $W_1$ accounts for the distance over which spatial mass must be transported and is a more natural way of comparing such equal-mass measures. We now state the key regularity condition on $r(\cdot,\z)$. 
\begin{assumption}\label{Wlip}
	For every $\z\in\cA$, the function $r(\cdot,\z)$ is $L$-Lipschitz continuous on $\cS$ with respect to the $W_1$ metric, i.e., 
	\begin{align}
		|r(\lambda,\z) - r(\lambda',\z)| &\leq LW_1(\lambda,\lambda'), & \lambda,\lambda'\in \cS.\label{Wlipeq}   
	\end{align}
\end{assumption}
The continuity assumption is natural for many performance metrics such as coverage and network capacity. Assumption \ref{Wlip} will later be verified for the coverage probability of a downlink cellular network in which the BSs follow an IPPP.

Assumption \ref{Wlip}, together with Jensen's inequality for the absolute value, implies that $R$ is also $L$-Lipschitz continuous with respect to $W_1$ since 
\begin{align}
	\abs{R(\lambda)-R(\lambda')}  &= \abs{ \Ex{r(\lambda, \z) - r(\lambda', \z)}} \label{rlip0}\\
	&\leq \Ex{\abs{r(\lambda, \z) - r(\lambda', \z)}} \label{rlip1}\\
	&\leq LW_1(\lambda,\lambda').\label{rlip2}
\end{align}
We also have the following lemma. 
\begin{lemma}\label{edlem}
    Under Assumption \ref{Wlip}, and if the partitions satisfy
    \begin{align}
    \max_{1\leq k\leq D}\sup_{\x\in\cA_k}
    \norm{\x-\x_k}
    \leq c_{\mathrm g}\sqrt{\frac{A}{D}},\label{cg}
    \end{align}
    where $\x_k \in \cA_k$ is the centroid of $\cA_k$ and $c_{\mathrm g}$ is a geometric constant independent of $D$, we have the bound
    \begin{align}
    E_\d\leq 2c_{\mathrm g}NL\sqrt{\frac{A}{D}}.
    \end{align}
\end{lemma}
\begin{IEEEproof}
	Using \eqref{rlip0}-\eqref{rlip2} in \eqref{edbound1}, we obtain
	\begin{align}
		E_\d &\leq LW_1(\lambda^\star,\lambda_\d^\star). \label{Ed1}
	\end{align}
	and therefore it suffices to bound $W_1(\lambda^\star,\lambda_\d^\star)$. Let us denote 
	\begin{align}
		w(\varphi) &:=  \int_{\cA} \varphi(\x)\delta\lambda^\star(\x) d\x = \sum_{k=1}^D \int_{\cA_k} \varphi(\x)\delta{\lambda}^\star(\x) d\x
	\end{align}
	where $\delta{\lambda}^\star:=\lambda^\star - \lambda_\d^\star$. Since $\int_{\cA_k} \delta{\lambda}^\star(\x) d\x = 0$, it holds that
	\begin{align}
		\abs{w(\varphi)} &= \abs{\sum_{k=1}^D \int_{\cA_k} (\varphi(\x)-\varphi(\x_k))\delta{\lambda}^\star(\x) d\x} \\
        &\leq \sum_{k=1}^D \int_{\cA_k} \abs{\varphi(\x)-\varphi(\x_k)}\abs{\delta{\lambda}^\star(\x)} d\x \\
		&\leqtext{$\varphi \in \cL_1$} \sum_{k=1}^D \int_{\cA_k} \norm{\x-\x_k}\abs{\delta{\lambda}^\star(\x)} d\x \\
		&\hspace{-1cm}\leqtext{\eqref{cg}} c_{\mathrm g}\sqrt{\tfrac{A}{D}}
        \sum_{k=1}^D\int_{\cA_k}
        \abs{\lambda^\star(\x)-\lambda_\d^\star(\x)}d\x \leq 2c_{\mathrm g}N\sqrt{\tfrac{A}{D}}, \label{wphi}
	\end{align}
    where \eqref{wphi} follows from the triangle inequality and the fact that $\int_{\cA}\lambda^\star(\x)d\x = \sum_k \int_{\cA_k} \lambda^\star(\x)d\x = N$ and likewise for $\lambda^\star_\d$. 
    Since this bound holds for every $\varphi\in\cL_1$, it follows that
\begin{align}
	W_1(\lambda^\star,\lambda_\d^\star)
	= \sup_{\varphi\in\cL_1} w(\varphi)
	\leq \sup_{\varphi\in\cL_1}\abs{w(\varphi)}
	\leq 2c_{\mathrm g}N\sqrt{\tfrac{A}{D}}.
\end{align}
Combining this with \eqref{Ed1} gives the required bound.
\end{IEEEproof}
The geometric condition in Lemma \ref{edlem} is satisfied by many standard partitions. For instance, if the partition is a regular polygon with at least 3 sides, the condition holds with $c_\mathrm{g} = 1$. More generally, the condition holds for any shape-regular family of partitions whose maximum centroidal radius scales as $\O{\sqrt{A/D}}$. It also accommodates nonuniform and adaptive meshes, such as those considered in \cite{mcrae2018optimal}, provided that their maximum centroidal radius satisfies the stated bound.

Intuition for the discretization error bound comes from viewing the replacement of $\lambda^\star$ by $\lambda_\d^\star$ as a \emph{mass transport} problem. The piecewise constant approximation only redistributes BS mass within each sub-region, and the maximum displacement from any point within a sub-region to its centroid is bounded by $c_{\mathrm{g}}\sqrt{A/D}$. Since the total intensity mass is $N$, the resulting $W_1$ distance, and hence the discretization error, scales linearly with $N$ and with the characteristic cell length $\sqrt{A/D}$, so that $E_\d=\mathcal{O}(NL\sqrt{A/D})$. The shape of the sub-regions affects the geometric constant $c_{\mathrm g}$, but not the scaling of the bound with $D$ or $N$.

Finally, the dependence on $D$ can be improved by making additional but unverifiable assumption on $\lambda^\star$ as stated in the following remark. 
\begin{rem}\label{ed2}
	If the optimal normalized intensity $\mu^\star = \tfrac{1}{N}\lambda^\star$ has bounded variation, i.e., $\int_{\cA} \norm{\nabla \mu^\star(\x)}d\x \leq B_\star$, then it holds that $E_\d \leq \frac{LNAB_\star}{D}$, as shown in Appendix \ref{ed2proof} in the supplementary material. However, since $\lambda^\star$ is itself unknown, it is not possible to verify if it has bounded variation or to obtain $B_\star$. Therefore, this improved rate cannot be used as an a priori bound without additional information about $\lambda^\star$. 
\end{rem}

\subsection{Bound on the sampling error}
Since $\lamt_\d$ maximizes $R$ over $\Lam$ and $\lamh_\d\in\Lam$, it follows that $R(\lamt_\d)-R(\lamh_\d)\geq 0$ so that 
\begin{align}
    E_s=\Ex{R(\lamt_\d)-R(\lamh_\d)}. \label{es2}
\end{align}
The term within the expectation can be decomposed as:
\begin{align}
	R(\lamt_\d ) - R(\lamh_\d ) =& \left(R(\lamt_\d ) - \Rh(\lamt_\d )\right) + \left(\Rh(\lamt_\d ) - \Rh(\lamh_\d )\right) \nonumber\\
	&+ \left(\Rh(\lamh_\d ) - R(\lamh_\d )\right).\label{rdbound}
\end{align}
Since $\lamh_\d  = \arg\max_{\lam_\d\in \Lam}\Rh(\lam_\d )$ and $\lamt_\d  \in \Lam$, it follows that $\Rh(\lamh_\d ) \geq \Rh(\lamt_\d )$ and we can drop the second term on the right. Therefore, the difference can be bounded as
\begin{align}
	R(\lamt_\d ) - R(\lamh_\d ) \leq& \max_{\lam_\d\in \Lam}\left(R(\lam_\d ) - \Rh(\lam_\d )\right)  \nonumber\\
	&+ \max_{\lam_\d \in\Lam}\left(\Rh(\lam_\d ) - R(\lam_\d )\right). \label{esb1}
\end{align}
Observe that the two terms in \eqref{esb1} measure, in either direction, the worst-case difference between the exact objective $R$ and its empirical estimate $\Rh$ over all $\lam_\d\in\Lam$. To bound these terms, we use the \emph{Rademacher complexity} of the class of performance metrics $\{r(\lam_\d,\cdot)\mid \lam_\d\in\Lam\}$, defined as
\begin{align}
	R_M := \Ex{\max_{\lam_\d\in \Lam} \frac{1}{M}\sum_{m=1}^M \varepsilon_m r(\lam_\d,\z_m)}
\end{align} 
where $\{\varepsilon_m\}$ are i.i.d. random variables drawn from the Rademacher distribution, i.e., $\mathbb{P}(\varepsilon_m = 1) = \mathbb{P}(\varepsilon_m = -1) = 1/2$. The expectation is with respect to both $\{\varepsilon_m\}$ and $\{\z_m\}$. In the context of learning, the Rademacher complexity measures how well a particular hypothesis class can fit to random labels. In the present setting, the relevance of $R_M$ is made precise by the following lemma.
\begin{lemma}\label{rad}
	It holds that
	\begin{align}
		\Ex{\max_{\lam_\d\in \Lam}\left(R(\lam_\d ) - \Rh(\lam_\d )\right)}  &\leq 2R_M \label{radeq1}\\
		\Ex{\max_{\lam_\d \in\Lam}\left(\Rh(\lam_\d ) - R(\lam_\d )\right)} & \leq 2R_M
	\end{align}
	where the expectations are with respect to the random user locations $\{\z_m\}$ appearing in the definition of $\Rh$. 
\end{lemma}
The proof of Lemma \ref{rad} is provided in Appendix \ref{radproof} of the supplementary material and follows along the lines of similar results in statistical learning theory; see e.g. \cite[Sec. 4.5]{bach2024learning}. A small value of $R_M$ implies a small expected uniform deviation between $\Rh(\lam_\d)$ and $R(\lam_\d)$ over $\lam_\d\in\Lam$. Using Lemma \ref{rad} in \eqref{es2} and \eqref{esb1}, we obtain
	\begin{align}
		E_s &\leq \Ex{\max_{\lam_\d\in \Lam}\left(R(\lam_\d ) - \Rh(\lam_\d )\right)}  \nonumber\\
		&+ \Ex{\max_{\lam_\d \in\Lam}\left(\Rh(\lam_\d ) - R(\lam_\d )\right)} \leqtext{Lemma \ref{rad}} 4R_M. \label{esproof1}
	\end{align}
It therefore suffices to bound $R_M$. Since Assumption \ref{Wlip} characterizes the sensitivity of $r$ in terms of the $W_1$ distance, whereas $R_M$ is indexed by the finite-dimensional vector $\lam_\d$, we first relate changes in $r$ to the Euclidean distance between discretized intensity vectors.
\begin{lemma}\label{rdlip}
	It follows from Assumption \ref{Wlip} that
	\begin{align}
		\abs{r(\lam_\d,\z) - r(\lam_\d',\z)} \leq \frac{L A_d A}{\sqrt{D}} \norm{\lam_\d-\lam_\d'}_2\label{rdlipeq}
	\end{align}
    for all $\lam_\d,\lam_\d'\in\Lam$ and $\z\in\cA$.
\end{lemma}
\begin{IEEEproof}
From \eqref{Wdef}, we have: 
\begin{align}
	W_1(\lambda_\d, \lambda_\d') &= \sup_{\varphi \in \cL_1}\int_{\cA} \varphi(\x)(\lambda_\d(\x)-\lambda'_\d(\x))d\x \\
	&\hspace{-1cm}=\sup_{\varphi \in \cL_1}\int_{\cA} (\varphi(\x)-\varphi(\x'))(\lambda_\d(\x)-\lambda'_\d(\x))d\x
\end{align}
since $\int_{\cA}\lambda_\d(\x)d\x = \int_{\cA}\lambda'_\d(\x)d\x$ for any fixed $\x' \in \cA$. Since $\varphi \in\cL_1$, we obtain
\begin{align}
	W_1(\lambda_\d,\lambda'_\d) &\leq \int_{\cA}\norm{\x - \x'} \abs{\lambda_\d(\x)-\lambda'_\d(\x)}d\x \\
	&\hspace{-5mm}\leq \text{diam}(\cA) \sum_{k=1}^D \int_{\cA_k} \abs{\lambda_{\d k}-\lambda'_{\d k}}d\x \\
	&\hspace{-5mm}= A_d \tfrac{A}{D}\norm{\lam_\d-\lam_\d'}_1 \leq A_d\tfrac{A}{\sqrt{D}}\norm{\lam_\d-\lam_\d'}_2
\end{align}
which yields the required Lipschitz constant.
\end{IEEEproof}

Accordingly, let $L_\d:=LA_dA/\sqrt{D}$ denote the resulting Euclidean Lipschitz constant. A key ingredient for bounding $R_M$ is the so-called vector contraction bound, first established in \cite{maurer2016vector}, which is re-stated below using the current notation: 
\begin{align}
	R_M \leq \sqrt{\frac{\pi}{2}}\frac{L_\d}{M}\EE\left[\max_{\lam_\d\in \Lam} \sum_{i=1}^M\sum_{j=1}^D g_{ij}\lambda_{\d j}\right] \label{vecon}
\end{align}
where $g_{ij}$ are i.i.d. standard normal random variables. The vector contraction inequality is stated in \cite{maurer2016vector} using Rademacher random variables, with the constant $\sqrt{2}$. The same argument applied to standard normal random variables yields \eqref{vecon}, with the constant $\sqrt{\pi/2}$, as shown in the proof of \cite[Proposition 1]{maurer2016vector}. The following lemma gives the resulting sampling-error bound.
\begin{lemma}\label{eslem}
	Under Assumption \ref{Wlip}, we have the bound
	\begin{align}
		E_s \leq 4NL A_d\sqrt{\frac{\pi D\log(D)}{M}}\label{esbound}
	\end{align}
\end{lemma}
\begin{IEEEproof}
    Substituting $L_\d$ from Lemma \ref{rdlip} and using the fact that $\sum_{i=1}^M g_{ij} = \eta_j \sim \cN(0,M)$ in \eqref{vecon}, we obtain
	\begin{align}
		R_M \leq \frac{LA_dA}{M}\sqrt{\frac{\pi }{2D}}\EE\left[\max_{\lam_\d\in \Lam} \sum_{j=1}^D \eta_j \lambda_{\d j} \right]
	\end{align} 
	Since $\sum_j \eta_j\lambda_{\d j} \leq \left(\max_j{\eta_j}\right)\sum_{j=1}^D \lambda_{\d j} = \frac{ND}{A}\max_j \eta_j$, it follows from \cite[Lemma 2.3]{massart2007concentration} that 
	\begin{align*}
		R_M &\leq \tfrac{LA_dA}{M}\sqrt{\tfrac{\pi }{2D}}  \tfrac{ND}{A} \sqrt{2M\log(D)} = NLA_d\sqrt{\tfrac{\pi D \log(D)}{M}}.
	\end{align*}
	Substituting into \eqref{esproof1}, we obtain the desired result. 
\end{IEEEproof}

In summary, the bound on the sampling error is $\Ot{\sqrt{D/M}}$ where the $\cOt$ notation hides the logarithmic dependencies as well as other constants. In this bound, the $\O{1/{\sqrt{M}}}$ dependence reflects the improvement obtained from averaging
over additional user samples. Interestingly, increasing $D$ has two opposing effects:  the Euclidean Lipschitz constant $L_\d$ decreases as $\O{1/{\sqrt{D}}}$ while the Gaussian complexity term associated with $\Lam$ increases as $\Ot{D}$. Thus, finer discretization reduces the discretization error but requires more user samples to bound the sampling error.

\subsection{Bound on excess error}
Combining the bounds in Lemmas \ref{edlem} and \ref{eslem}, we obtain the total bound on the excess error for square cells as 
\begin{align*}
	E &\leq NL\sqrt{\tfrac{2A}{D}} +  4NLA_d\sqrt{\tfrac{\pi D\log(D)}{M}} =\cOt\left(\tfrac{1}{\sqrt{D}} + \sqrt{\tfrac{D}{M}}\right).
\end{align*}
where the $\cOt$ notation hides the constants including $N$, $A$, and $A_d$. As with most worst-case generalization bounds, the expression is informative when its right-hand side is small relative to the range of the performance metric. In particular, the bound reveals how the excess error scales with $D$, $M$, and $N$, and clarifies the tradeoff between its discretization and sampling components.  Roughly, the error bound would be minimum if the two terms match, i.e., when $D^2 \approx M$. Hence, for fixed $N$, $A$, and $A_d$, if we want to ensure that the excess error is at most $\epsilon$, then we should choose $D = \cOt\left(\frac{1}{\epsilon^2}\right)$ and $M = \cOt\left(\frac{1}{\epsilon^4}\right)$, ignoring the logarithmic terms. We refer to these bounds as discretization and sampling complexities of the proposed scheme, respectively. If the underlying intensity has bounded variation (cf. Remark \ref{ed2}), then the complexities improve to $D = \cOt\left(\frac{1}{\epsilon}\right)$ and $M = \cOt\left(\frac{1}{\epsilon^3}\right)$. 

Thus, the proposed formulation learns a spatial BS intensity directly from user-location samples without requiring prior knowledge or parametric estimation of $P_\z$, while the bound quantifies the number of samples needed to support a chosen spatial resolution.

\section{Optimal Deployment of BSs}\label{optim}
We now present the optimization approach to solve \eqref{pe} and discuss how the solution provides soft guidelines along with their connection to the real-world BS-deployment.

\subsection{Soft Deployment: Solving \texorpdfstring{\eqref{pe}}{(pe)}}
Since \eqref{pe} is high-dimensional, we obtain the soft-deployment using only first-order methods. Projected gradient methods are less attractive because maintaining feasibility at every iteration entails solving a Euclidean projection subproblem over the scaled simplex $\Lam$, typically using a dedicated sorting- or active-set-based procedure \cite{condat2016fast}. Instead, we propose to use the mirror ascent algorithm that yields a closed-form update. Starting at an arbitrary $\lam^{(0)}$ in the relative interior of $\Lam$ (such as the point $\lambda_{\d i}^{(0)} = N/A$ for all $i$), the mirror ascent updates at the $t$-th iteration are evaluated as
\begin{align}\label{mirup}
	\lam_\d^{(t+1)} = \arg\max_{\lam\in \Lam} \ip{\nabla \Rh(\lam_\d^{(t)})}{\lam} - \frac{1}{\eta^{(t)}}D_{\text{KL}}(\lam,\lam_\d^{(t)})
\end{align}
where $\eta^{(t)}$ is the step-size and $D_{\text{KL}}(\lam,\lam_\d^{(t)}):=\sum_{i=1}^D \lambda_i \log(\lambda_i/\lambda_{\d i}^{(t)})$ is the KL divergence. These updates can be written in closed form as in \cite{fatkhullin2024taming}.
The empirical form of $\Rh$ further motivates the use of the stochastic mirror ascent (SMA) algorithm:
\begin{align}\label{sma}
	\lambda_{\d i}^{(t+1)} = \frac{ND}{A}\frac{\lambda_{\d i}^{(t)}\exp\left(\eta^{(t)}g^{(t)}_i\right)}{\sum_{j=1}^D\lambda_{\d j}^{(t)}\exp\left(\eta^{(t)}g^{(t)}_j\right)}
\end{align}
for $1 \leq i \leq D$, where $\{g^{(t)}_i\}$ are entries of 
\begin{align}\label{sgdef}
	\g^{(t)}:=\frac{1}{\abs{\cM_t}}\sum_{m\in\cM_t}\nabla r(\lam_\d^{(t)},\z_{m})	
\end{align}
with $\cM_t$ being a random mini-batch of users at iteration $t$. 

After running for $T$ iterations, the optimization algorithm may yield a point $\blamdt$. Without additional global structure, such as concavity of $\hat{R}$, it is not possible to find a bound on the optimization error $E_o = \EE[\Rh(\lamh_\d )-\Rh(\blamdt)]$ that decreases to zero with $T$. If, however, we assume that $r(\cdot,\z)$ is relatively smooth, which follows if  $\norm{\nabla r(\lam_\d,\z)- \nabla r(\lam_\d',\z)}_{\infty}\leq L_{\nabla}\norm{\lam_\d-\lam_\d'}_1$, then existing results for nonconvex stochastic mirror ascent guarantee convergence to a stationary point of \eqref{pe} \cite{fatkhullin2024taming}. Further, Appendix \ref{sec:noise_limited} of the supplementary material shows that under the noise-limited IPPP approximation, the coverage probability is concave and yields the bound $E_o \leq N\sqrt{\frac{2\log D}{T}}$ from \cite{nemirovski2009robust}. 

Finally, since the empirical problem is generally nonconcave, the stationary point reached by mirror ascent may depend on the initialization. We initialize the BS intensity using the empirical user-density map, since it provides a data-driven starting point that already allocates greater deployment mass to regions of high user concentration. 

\subsection{From Guidelines to BS Deployment}\label{sec-deploy}
As the optimized intensity function $\blamt(\x) := \sum_{k=1}^D \bar{\lambda}_{\d k}^T \ind_{\x \in \cA_k}$ serves as a guideline for BS deployment, we use it to generate BS deployment by sampling the underlying spatial point process with optimized intensity. Here, each realization serves as a potential deployment candidate with different realized performance. To further optimize the deployment, we generate several candidates and select the one with the highest estimated performance using an independent set of user-location samples. For a realization $\phi$ of the spatial point process $\Phi_\lambda$, let $r(\phi,\z)$ denote, with a slight abuse of notation, the performance of the corresponding BS deployment for a user at $\z$, and define
$R(\phi):=\EE_{\z}[r(\phi,\z)]$. We then have 
\begin{align}
R(\lambda)
=\EE_{\phi}[R(\phi)]
=\EE_{\phi}\EE_{\z}[r(\phi,\z)]. \label{Rexp}
\end{align}
where the outer expectation is over realizations $\phi$ of $\Phi_\lambda$. 

We draw $K$ independent candidate deployments $\{\phi_k\}_{k=1}^K$ from the point process $\Phi_{\bar{\lambda}^T}$. Note that the original user-location samples $\cZ$ that were used to obtain $\blamdt$ should not also be used to select the final realization. Since $\blamdt$ was chosen specifically to perform well at these locations, reusing them may favor a deployment that performs well on $\cZ$ but performs poorly for the underlying user population. We therefore draw an independent hold-out validation set of user locations $\cZ_v: = \{\z_m^v\}_{m=1}^{M_v}$, which is used only to select the final deployment $\phihk = \phi_{\hat{k}}$, where
\begin{align}
    \hat{k} = \arg\max_{1 \leq k \leq K} \Rh(\phi_k, \cZ_v) := \frac{1}{M_v}\sum_{m=1}^{M_v} r(\phi_k,\z_m^v).\label{phihk}
\end{align}
In the next subsection, we show that for sufficiently large values of $K$ and $M_v$, the selected deployment cannot be much worse than the average performance $R(\blamt)$ of the optimized intensity function. 

\subsubsection{Placement complexity} Define the placement error as
\begin{align}
    E_\p = \EE\left[(R(\blamt) - R(\phihk))_+\right]
\end{align}
where the positive part $(a)_+ = \max\{0,a\}$ is used to exclude the realizations that perform better than the average. The following lemma quantifies the additional error incurred in obtaining a concrete deployment.
\begin{lemma}\label{eplem}
    If $r(\phi,\z) \in [0,B_r]$ for all $\phi$ and $\z$, then the placement error satisfies
    \begin{align}
        E_\p \leq \tfrac{B_r}{K} + B_r\sqrt{\tfrac{\log(K)}{2M_v}}
    \end{align}
\end{lemma}
\begin{IEEEproof}
    The placement error can be split as
    \begin{align*}
        E_\p \leq & \EE[(R(\blamt) - \max_k R(\phi_k))_+] + \Ex{\max_k R(\phi_k) - R(\phihk)}
    \end{align*}
    where the two terms capture, respectively, the errors due to generating only $K$ candidate deployments and selecting $\phihk$ in \eqref{phihk} using the validation estimate $\Rh$ instead of $R$. 

    For the first term, observe that conditioned on $\blamt$, the random variables $\{R(\phi_k)\}_{k=1}^K$ are i.i.d. in $[0,B_r]$ with mean $\EE_{\phi_k}{R(\phi_k)} = R(\blamt)$. Hence, using Markov's inequality on the random variable $B_r - R(\phi_k)$ we obtain the required bound:
    \begin{align}
        \EE\left(R(\blamt)-\max_k R(\phi_k)\right)_+ \leq \tfrac{B_r}{K}\label{ep1bound}
    \end{align}
    which is proved in Appendix \ref{app_ep1} of the supplementary material. 

    In the second term, conditioned on the candidate deployments $\{\phi_k\}_{k=1}^K$, each $\Rh(\phi_k, \cZ_v)$ is the average of $M_v$ i.i.d. random variables taking values in $[0,B_r]$. Therefore, Hoeffding's lemma implies that $\Delta\Rh_k := \Rh(\phi_k, \cZ_v) - \EE_{\cZ_v}[\Rh(\phi_k, \cZ_v)]$ is sub-Gaussian with parameter $B_r^2/(4M_v)$ \cite[Sec. 1.2.1]{bach2024learning}. Hence, an expectation-of-the-maximum bound for sub-Gaussian random variables yields \cite[Prop. 1.5]{bach2024learning}: 
    \begin{align}
        \EE[\max_k \Delta\Rh_k] \leq B_r\sqrt{\tfrac{\log K}{2M_v}}. \label{expmax}
    \end{align}
    Let $k^\star = \arg\max_k R(\phi_k)$ so that $\max_kR(\phi_k) = R(\phi_{k^\star})$. Since $\phihk$ maximizes $\Rh$, we have
    \begin{align*}
    R(\phi_{k^\star})-R(\phihk)
    &\leq R(\phi_{k^\star})-\Rh(\phi_{k^\star}, \cZ_v)
    +\Rh(\phihk, \cZ_v)-R(\phihk).
    \end{align*}
    Since $k^\star$ depends only on $\{\phi_k\}$, the first term on the right has zero mean. The second term is bounded by $\EE[\max_k \Delta\Rh_k]$. Therefore from \eqref{expmax}, we obtain
\begin{align}
\EE\left[\max_{1\leq k\leq K}R(\phi_k)-R(\phihk)\right]
\leq B_r\sqrt{\frac{\log K}{2M_v}}
\end{align}
which, together with \eqref{ep1bound}, yields the desired bound.     
\end{IEEEproof}

Lemma \ref{eplem} shows that $E_\p\leq\epsilon$ can be achieved with $K=\O{1/\epsilon}$ candidate deployments and $M_v=\Ot{1/\epsilon^2}$ validation samples. Since sampling candidate deployments from $\Phi_{\blamt}$ is usually inexpensive, the validation sample complexity is generally the more demanding requirement. 

\begin{algorithm}[htb]
\caption{End-to-End Base Station Deployment}
\label{algo_e2e}
\begin{algorithmic}[1]
\State Partition $\cA$ into $D$ subregions and construct $\Lam$
\State Obtain $M$ user-location samples $\cZ$
\State Construct the empirical objective $\Rh$
\State Run updates in \eqref{sma} for $T$ iterations to obtain $\blamdt$
\State Generate $K$ candidate deployments $\{\phi_k\}_{k=1}^K$ from $\Phi_{\blamt}$
\State Obtain $M_v$ validation user-location samples
$\cZ_v$
\State Solve \eqref{phihk} and return $\phihk$
\end{algorithmic}
\end{algorithm}

\subsubsection{End-to-end performance guarantee} Algorithm \ref{algo_e2e} summarizes the complete procedure for obtaining a concrete deployment from user-location samples. The following theorem combines the earlier bounds on discretization, sampling, optimization, and placement. 
\begin{theorem}\label{thm:e2e}
Under the conditions of Lemmas \ref{edlem}, \ref{eslem}, and \ref{eplem}, the deployment returned by Algorithm \ref{algo_e2e} satisfies
\begin{align}
\EE&\left[\left(R(\lambda^\star)-R(\phihk)\right)_+\right] \nonumber\\
&\leq E_o + \Ot{\sqrt{\tfrac{A}{D}} + \sqrt{\tfrac{D}{M}} + \tfrac{1}{K} + \sqrt{\tfrac{1}{M_v}}}.
\label{e2ebound}
\end{align}
\end{theorem}
The result follows from the decomposition
\begin{align}
\left(R(\lambda^\star)-R(\phihk)\right)_+
&\leq R(\lambda^\star)-R(\tilde{\lambda}_\d)
+R(\lamt_\d)-R(\blamdt) \nonumber\\
&+\left(R(\blamt)-R(\phihk)\right)_+\end{align}
where the first two differences are nonnegative by the optimality of $\lambda^\star$ and $\lamt_\d$, respectively. Applying the decomposition in \eqref{rdbound} with $\blamdt$ in place of $\lamh_\d$, and using $\Rh(\lamt_\d)\leq\Rh(\lamh_\d)$, gives $\EE[R(\lamt_\d)-R(\blamdt)]\leq 4R_M+E_o$. Taking expectations and applying Lemmas \ref{edlem}, \ref{eslem}, and \ref{eplem}, together with the definition of $E_o$, proves \eqref{e2ebound}. Except for the numerical optimization error $E_o$, the other terms can be made arbitrarily small by setting $D = M_v = \Ot{1/\epsilon^2}$, $M = \Ot{1/\epsilon^4}$, and $K = \Ot{1/\epsilon}$. As stated earlier, $E_o$ can be made small in certain cases, such as when the deployment process is IPPP and the system is noise-limited. 

Beyond BS deployment, Theorem \ref{thm:e2e} applies to intensity-based spatial design whenever the point-process model and performance metric satisfy the stated regularity conditions. To the best of our knowledge, it provides the first end-to-end nonasymptotic error decomposition that incorporates discretization, finite-sample learning, and realization selection components, up to the numerical optimization error $E_o$.

\section{Optimal Deployment under Specific Performance Metrics}
We now specialize the framework to an inhomogeneous Poisson point process (PPP) $\Phi_\lambda$ on $\cA$ with intensity $\lambda$. In particular, we consider performance of a downlink cellular network at conventional sub-6 GHz frequencies under two cell association rules: instantaneous signal-to-interference-plus-noise ratio (SINR)-based association and average SINR-based association. We consider an orthogonal multiple-access downlink system model, where the BSs transmit constantly with a fixed power $P$ to a single desired mobile user on any particular time-frequency resource. As in \cite{haenggi2012stochastic}, we consider a generic path-loss function $\ell(r)$ that satisfies (a) $0 < \ell(r) \leq 1$, (b) $\ell'(r) \leq 0$, and  (c) $\int_{0}^{\infty} \ell(r)r dr  < \infty$, e.g.,  $\ell(r) = \frac{1}{1+r^\alpha/c}$ where $\alpha \in (2,5]$ is the path-loss exponent. For our analysis, we additionally require that the path-loss function does not decay faster than an exponential, i.e., there exists $L_\ell < \infty$ such that $\ell(r)e^{L_\ell r}$ becomes an increasing function of $r$.  In practice, path-loss is typically modeled using a power law rather than exponential or super-exponential functions, so $L_\ell$ tends to be small across this class of models. The average received power at the user is given by $P\ell(\norm{\x-\z})$ and the corresponding average SNR is $P\ell(\norm{\x-\z})/\sigma^2$ where $\sigma^2$ is the noise power. The interference power received by a user located at $\z$ is denoted by $I(\z)$. We consider Rayleigh fading, where the fading gain for each link is i.i.d. exponential with unit mean, in which case, the Laplace transform of $I(\z)$ is known to be \cite[Sec. III-C]{andrews2016primer}:
\begin{align}
	\cL_{I(\z)}(s) = \exp\Big(-\int_{\cR}\tfrac{\lambda(\y)\d\y}{1+\tfrac{1}{sP\ell(\norm{\y-\z})}}\Big)\label{LI}
\end{align}
where $\cR$ is the (sub-)region in $\cA$ where the interfering BSs lie. As we shall see later, the expression allows us to calculate the expectation of various functions of $I(\z)$. In particular, the signal-to-interference-plus-noise ratio (SINR) \cite{haenggi2009stochastic} at the user located at $\z$ from a BS located at $\x$ is given by 
\begin{align}
	\mathrm{SINR}_{\x} = \tfrac{GP\ell(\norm{\x - \z})}{I(\z) + \sigma^2}.
\end{align}
The mobile user located at $\z \in \cA$ is said to be in coverage if the BS (located at $\x$) selected under the association rule provides an SINR at $\z$ exceeding the threshold $\tau$ . We are interested in the performance metric coverage probability which captures the probability that a typical user is covered, i.e., $p_c(\lambda)=E_z[p_c(\lambda,\z)]= \mathbb{P}[\text{SINR}(\z)\geq \tau]$. 

\subsection{Instantaneous SINR-Based Association}\label{sec-ia}
In this association model, each mobile user connects to the BS offering the highest instantaneous SINR, considering both fading and path-loss. Interestingly, a user may not always be associated with the closest BS, since a higher instantaneous SINR may be received from a different BS. Let us denote the SINR for the user at $\z$ when it connects to a BS at $\x$ by $\text{SINR}_{\x}(\z)$, or simply by $\text{SINR}_\x$. Hence, the probability that a user at $\z$ is covered is given by 
\begin{align}
	p_c(\lambda,\z) = \Px{\max_{\x \in \Phi_\lambda} \text{SINR}_{\x} \geq \tau}\label{pci}
\end{align}
where the probability is calculated with respect to the random fading gains and the random BS point process. In order to simplify \eqref{pci}, we assume that $\tau > 1$ so that for a given user, the instantaneous SINR can exceed $\tau$ for at most a single BS, as in \cite{andrews2016primer}. Subsequent applications of the Campbell-Mecke as well as Slivnyak's theorems \cite{baccelli2010stochastic} yield
\begin{align}
	p_c(\lambda,\z) = \int_{\cA}\Px{\mathrm{SINR}_{\x} > \tau}\lambda(\x) \d\x. \label{pci2}
\end{align}
Since  $G \sim \text{Exp}(1)$, we have that 
\begin{align}
	&\Px{\mathrm{SINR}_{\x} > \tau} = \EE_{I(\z)}\left[\Px{\mathrm{SINR}_{\x} > \tau\mid I(\z)}\right] \\
	&=\EE_{I(\z)}\left[\exp\left(-\tau \frac{\sigma^2+I(\z)}{P\ell(\norm{\x-\z})}\right)\right] \\
	&=\exp\left(-\frac{\tau \sigma^2}{P\ell(\norm{\x-\z})}\right)\cL_{I(\z)}\left(\frac{\tau}{P\ell(\norm{\x-\z})}\right)\label{iSINR1}
\end{align}
where the expression for $\cL_{I(\z)}(s)$ is given in \eqref{LI} with $\cR= \cA$. Substituting \eqref{iSINR1} into \eqref{pci2}, we obtain the required coverage probability expression
\begin{align}\label{insteq}
	p_c(\lambda,\z) &= \int_{\cA}\lambda(\x)e^{-\frac{\tau \sigma^2}{P\ell(\norm{\x-\z})}-\int_{\cA} \lambda(\y)/(1+\frac{\ell(\norm{\x-\z})}{\tau\ell(\norm{\y-\z})})\d\y}\d\x \nonumber\\
	&=\int_{\cA} \lambda(\x) h(\lambda,\x) d\x
\end{align}
where we have suppressed the dependence on $\z$ and defined
\begin{align}
	h(\lambda,\x)&:=\exp\left[ - a(\x) -  \int_{\cA} \lambda(\y)b(\x,\y)d\y\right]\\
	a(\x) &:= \frac{\kappa}{\ell(d_\x)} \hspace{1cm} b(\x,\y):= \frac{\tau \ell(d_\y)}{\tau \ell(d_\y)+\ell(d_\x)}. \label{abdef}
\end{align}
with $d_\x:= \norm{\x-\z}$, $d_\y:=\norm{\y-\z}$, and $\kappa := \tau \sigma^2/P$. The following lemma, whose proof is provided in Appendix \ref{lips-proof} in the supplementary, verifies Assumption \ref{Wlip} for \eqref{insteq}. 

\begin{lemma}\label{pclip}
The coverage probability in \eqref{insteq} under instantaneous SINR association rule is $\tfrac{2}{3}L_\ell$-Lipschitz continuous with respect to the $W_1$ metric.  
\end{lemma}

As an example, consider the path-loss $\ell(r) = \frac{1}{1+r^{\alpha}/c}$ for some $c > 0$. For this case, it can be seen that  
\begin{align}
	\sup_{r\geq 0}\tfrac{\abs{\ell'(r)}}{\ell(r)} \leq \left(\tfrac{(\alpha-1)^{\alpha-1}}{c}\right)^{1/\alpha} =: L_\ell
\end{align}
so that $p_c(\lambda,\z)$ is $\tfrac{2}{3}L_\ell$-Lipschitz. For $c = 1$ and $\alpha \in (2,5]$, this quantity is upper bounded by approximately 3. The expression in \eqref{insteq} can be discretized as demonstrated earlier. 

\subsubsection{Extension to multi-tier networks}
The current formulation also applies to a heterogeneous cellular network with $Q$ overlaid independent IPPP tiers of BSs, each of which may differ in terms of transmit power, deployment density, target SINR, etc. Denoting the intensity of the BSs in the $q$-th tier by $\lambda^{[q]}(\x)$, the coverage probability, with instantaneous SINR-based association, may similarly be derived as \cite{andrews2016primer}: 
\begin{align*}
	p_c(\{\lambda^{[q]}\},\z) &= \sum_{i = 1}^Q \int_{\cA} \lambda^{[i]}(\x)e^{-\tfrac{\tau_i\sigma^2}{P_i\ell(\norm{\x - \z})}-w_i(\x,\z)} \d\x \\
	w_i(\x,\z)&= \sum_{j=1}^Q \int_{\cA} \frac{\lambda^{[j]}(\y)\d\y}{1+\frac{P_i\ell(\norm{\x-\z})}{\tau_iP_j\ell(\norm{\y-\z})}}
\end{align*}
where $P_i$ and $\tau_i > 1$ are the powers and SINR thresholds for the $i$-th tier. Hence, the goal now is to maximize $\Ex{p_c(\{\lambda^{[q]}\},\z)}$ under the constraints under the prescribed mean BS counts for each tier. 
Discretizing these densities yields a $QD$-dimensional optimization problem. 

\subsection{Average SINR-Based Association}\label{sec-aa}
In this association model, the mobile user connects to the BS that offers the highest average SINR. Under the equal-power and monotone path-loss model considered here, this rule translates to associating each user with the nearest BS. For a user located at $\z$, let $\x_{\z}$ be the nearest BS and $R_{\z}:=\norm{\x_{\z}-\z}$ be the serving distance. In case the IPPP realization does not contain any BS, we set $R_\z = \infty$ and the corresponding event contributes zero coverage. Then, the coverage probability is the probability that the SINR at $\x_{\z}$ is greater than the threshold $\tau$, i.e.,
\begin{align}
	p_c^a(\lambda,\z) &= \Px{\mathrm{SINR}_{\x_{\z}} > \tau} \\
	&= \int_0^\infty \Px{\mathrm{SINR}_{\x_{\z}} > \tau \mid R_{\z} = r}f_{R_{\z}}(r)\d r \label{avSINR1}
\end{align}
where $f_{R_{\z}}(r)$ is the density of $R_{\z}$ for finite $r$.

To calculate $f_{R_{\z}}$, recall that for a PPP, the probability that no BS lies within the set $\cB(\z,r)=\{\x \in \cA \mid \norm{\x-\z}\leq r\}$ is given by 
\begin{align}
	\Px{R_{\z} > r} = \exp\left(-\int_{\cB(\z,r)} \hspace{-6mm}\lambda(\y)\d\y\right).
\end{align}
Therefore, for finite $r$, the density of $R_{\z}$ is given by $-\frac{\d}{\d r}\Px{R_{\z} > r}$, which evaluates to
\begin{align}
	f_{R_{\z}}(r) &=\exp\left(-\int_{\cB(\z,r)} \hspace{-6mm}\lambda(\y)\d\y\right) \frac{\d}{\d r}\left[\int_{\cB(\z,r)} \hspace{-6mm}\lambda(\y)\d\y\right]. \label{eq:fr_inhom}
\end{align}
The derivative in \eqref{eq:fr_inhom} can be evaluated analytically for simple regions or numerically for a general deployment region $\cA$.

Next, given $R_{\z}=r$, we have that $\mathrm{SINR}_{\x_{\z}} = \frac{GP\ell(r)}{I(\z) + \sigma^2}$. Hence the conditional probability becomes
\begin{align}
	&\Px{\mathrm{SINR}_{\x_{\z}} > \tau \mid R_{\z} = r} \nonumber\\
    &\hspace{2cm}= \EE_{I(\z)}\left[\Px{\mathrm{SINR}_{\x_{\z}} > \tau \mid R_{\z} = r, I(\z)}\right]\nonumber\\
	&\hspace{2cm}= \EE_{I(\z)}\left[\exp\left(-\tau \tfrac{\sigma^2 + I(\z)}{P\ell(r)}\right)\right] \\
	&\hspace{2cm}= \exp\left(-\tfrac{\tau \sigma^2}{P\ell(r)}\right)\cL_{I(\z)}\left(\tfrac{\tau}{P\ell(r)}\right)\label{avSINR2}
\end{align}
where the expression for $\cL_{I(\z)}(s)$ is given in \eqref{LI} with $\cR= \cA \setminus \cB(\z,r)$, since only BSs farther than the serving BS contribute interference. Hence, \eqref{avSINR2} becomes
\begin{align}\label{aveq}
	&\Px{\mathrm{SINR}_{\x_{\z}} > \tau \mid R_{\z} = r} \nonumber\\
	&=\exp\Big(-\frac{\tau \sigma^2}{P\ell(r)}-\int_{\cA\setminus \cB(\z,r)}\tfrac{\lambda(\y)\d\y}{1+\tfrac{\ell(r)}{\tau \ell(\norm{\y-\z})}}\Big)
\end{align}
which can be substituted into \eqref{avSINR1} to obtain the coverage probability. We remark that for other fading distributions as well, the distance-based association rule remains unchanged, and only the conditional probability in \eqref{avSINR2} changes. Moreover, $p_c^a(\lambda,\z)$ is $L_\ell/e$-Lipschitz
with respect to $W_1$, as proved in Appendix~\ref{lips-proof} of the supplementary material.

\subsection{Practical Applications}
We now present several additional practical applications of the proposed soft deployment framework, where the derived guidelines can be adapted to more realistic environments.

\subsubsection{Additional deployment constraints}
An important feature of the current framework is that it is relatively straightforward to incorporate physical constraints into the formulation. We discuss some common deployment constraints.

It is common to have certain areas to be off-limits for BS deployments, e.g., due to regulations or geographic limitations. Such \emph{exclusion zones} may be enforced by setting $\lambda_{\d k} = 0$ whenever $\cA_k$ is part of an exclusion zone. More generally, if the average number of BSs within a region $\cB \subset \cA$ consisting of complete grid cells is to be restricted to at most $N_B < N$, we can include the constraint
\begin{align}
	\sum\nolimits_{k: \cA_k \subset \cB} \lambda_{\d k} \leq \tfrac{N_BD}{A}.
\end{align} 
In contrast, certain high-priority regions may necessitate minimum deployment intensity requirements, which can be imposed by requiring $\lambda_{\d k} \geq \lambda_{\min}$. Finally, in multi-tier deployments, we might specifically impose constraints of the form $\sum_{q=1}^Q\lambda_{\d k}^{[q]} \leq \frac{D}{A}$ so as to ensure that each sub-region gets, on average, at most one BS of any tier. 

\subsubsection{Incremental deployment} Incremental deployment refers to the scenario in which a deterministic set $\Phi_0$ of $N_0$ BSs has already been installed and the objective is to optimize the intensity of an additional BS deployment having mean cardinality $N$. The new BS intensity $\lambda$ can help identify regions where further infrastructure is most beneficial given the coverage and interference provided by $\Phi_0$. The coverage expressions are modified to account for the superposition $\Phi_0 \cup \Phi_{\lambda}$ so that the fixed BSs participate in association and appear explicitly as interferers in the SINR expressions, while their locations remain unchanged during the optimization. The corresponding IA and AA coverage expressions are provided in Appendix~\ref{incremental} of the supplementary material.

\subsubsection{Adaptive deployment}
The proposed framework separates the deployment problem into two stages: soft deployment and concrete deployment, allowing additional site-specific constraints to be accommodated during the second stage. For instance, detailed site availability or deployment restrictions may only become known after the intensity has been optimized. If a small portion of a selected grid cell is unavailable, the BS can still be placed within the feasible part of that cell. If an entire cell becomes unavailable, its intensity mass may be redistributed among nearby feasible cells, allowing the deployment to be adjusted without solving the original optimization problem again.

\subsubsection{Sustainable deployment}
Green 6G targets can be imposed within the framework by attaching a cost density $c(\mathbf x)$ to every candidate site. The cost density may model location-specific carbon footprint, a proxy for electromagnetic exposure, or backhaul rental costs. Defining the expected aggregate cost $C(\lambda) := \int_{\cA} c(\x)\lambda(\x) d\x$, we can (a) include a soft-penalty and maximize $R(\lambda)-\beta C(\lambda)$ to explore the coverage-cost tradeoff, or (b) impose a budget on the expected cost of the form $C(\lambda)\leq C_{\max}$, which becomes a linear constraint after discretization. 

\subsubsection{Changing user demand}
The optimized intensity naturally assigns different deployment priorities to different regions. Examples include high-demand regions such as malls and stadiums. The framework can also be naturally extended to scenarios where the user distribution changes over time, by updating the user samples and re-optimizing the intensity.

\section{Numerical Results}\label{sec-results}
In this section, we evaluate the proposed approach with the instantaneous coverage probability derived in Sec. \ref{sec-ia}. A comparison of instantaneuous and average association scenarios is provided in Appendix \ref{sec:association_comparison}
of the supplementary material. We begin by visualizing the optimized intensity under different user distributions and deployment constraints (Sec. \ref{sec-spatial}). Next, we compare the performance of the proposed and baseline deployment approaches, and validate the theoretically predicted effects of discretization, finite number of user-location samples, and validation-based deployment selection (Sec. \ref{sec-theory}). Finally, we apply the proposed framework to find an optimal incremental BS deployment on top of an existing cellular deployment in the city of Chicago. 


\begin{table}[ht]
\centering
\caption{Common numerical parameters. Scenario-specific values
are stated in the corresponding subsections.}
\label{tab:numparams}
\footnotesize
\setlength{\tabcolsep}{3pt}
\renewcommand{\arraystretch}{1.05}
\begin{tabular}{lclc}
\hline
Parameter & Value & Parameter & Value\\
\hline
Deployment area $\cA$ & $10\times10$ km$^2$
    & Number of cells $D$ & $1600$\\
Path-loss exponent $\alpha$ & $3.5$
    & iterations $T$ & $100$\\
SINR threshold $\tau$ & $2$
    & Training samples $M$ & $1600$\\
Transmit power $P$ & $1$
    & Validation samples $M_v$ & $400$\\
Path-loss scale $c$ & $(0.2\,\mathrm{km})^\alpha$
    & Test samples $M_t$ & $10000$\\
Noise power $\sigma^2$ & $10^{-4}$
    & Candidates $K$ & $1000$\\
\hline
\end{tabular}
\end{table}

All experiments use the default parameters in Table~\ref{tab:numparams} unless otherwise noted. Along with the unconstrained (greenfield) deployments, we will also consider deployments under two types of  deployment constraints: exclusion zones and incremental deployments with existing BSs at known locations.  The code is provided in the supplementary material. 

The mirror-ascent updates in \eqref{mirup} are initialized at $\lambda^{(0)} = 0.99 \lambda_{u} + 0.01\lambda_{c}$ where $\lambda_u$ and $\lambda_c$ refer to the demand-proportional and uniform intensity functions, respectively. The small uniform component ensures a strictly positive initialization on the relative interior of the feasible simplex, while the demand-proportional component serves as a data-dependent warm start. For experiments involving concrete deployments, the training set is used to learn the intensity while an independent validation set is used to select a deployment, whose final performance is evaluated on a test set. 


\subsection{Optimized Intensity Maps: Guidelines and Trends}\label{sec-spatial}
We begin with several examples with different scenario/setups to demonstrate the application of the proposed soft deployment approach. The solutions of these examples  jointly illustrate how the optimized intensity changes with the user distribution and deployment constraints, providing important insights. In all figures, the color scale represents the spatial distribution of the BS intensity while the training user locations are represented by red dots. 
\begin{figure}[ht]
	\centering
	  \includegraphics[width=0.7\linewidth]{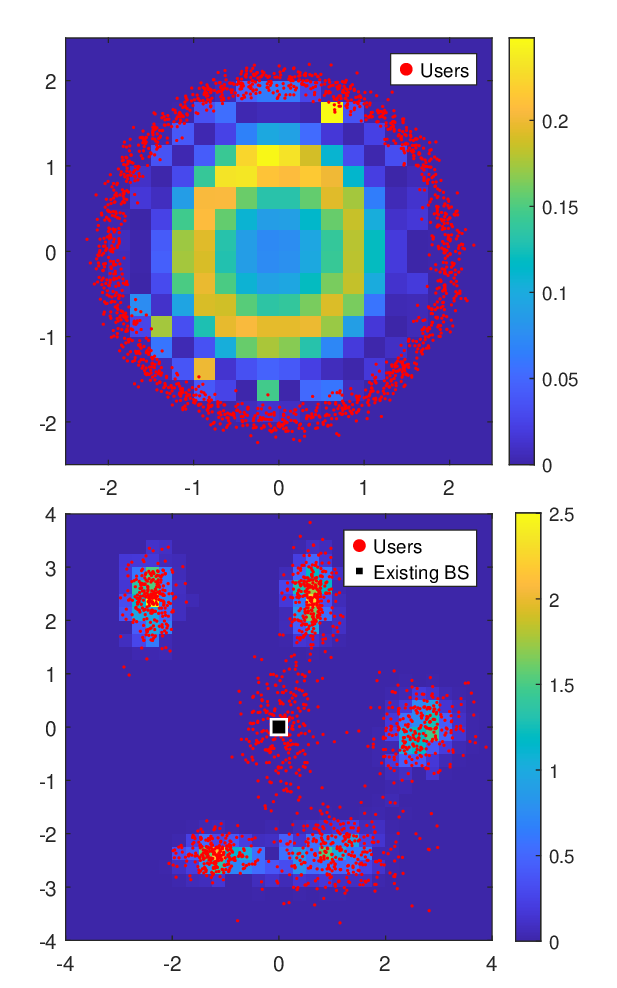}
	\caption{Optimized BS intensities for representative  user distributions. (top) Annular user-location distribution with mean $N=1$. (bottom) Six-cluster user-location distribution with one already installed BS and mean $N=5$ additional BSs.}
	\label{fig:basicmaps}
\end{figure}

We first consider a scenario with an annular user demand with $N = 1$ as shown in Figure~\ref{fig:basicmaps}(top). Observe that the optimized intensity does not simply follow the user density, but instead forms a concentric annulus of smaller radius. The inward displacement reflects a trade-off between remaining close to the individual users and providing comparable service around the entire ring. In additional runs, see Appendix~\ref{noise_sweep} of the supplementary material, we observed that the optimized annulus initially contracts towards the center but again expands towards the edges, as noise level decreases. In contrast, a hard single-BS deterministic optimizer places the BS at the center in the low-noise regime and switches to an off-center location near the user annulus as the noise increases, thereby masking the gradual radial transitions revealed by the optimized intensity. Further, due to radial symmetry of user distribution, for deterministic optimizer, any other location at the same distance from the center but at a different angle is also optimal. The optimizer picks one of these arbitrarily, again masking the equivalence of these points. Instead the soft deployment approach reveals the equivalence of points providing an additional flexibility to let us choose any BS location from these equivalent options. 

We next consider a scenario where the users form six spatial clusters, one of which is already served by an installed BS, over an $8 \times 8$ km$^2$ area, as shown in Figure~\ref{fig:basicmaps}(bottom). As expected, for $N = 5$, the optimized incremental-BS intensity is only concentrated primarily at the remaining five clusters. Thus, the intensity map also reveals how deployment priorities change with the installed network. Interestingly, the low-intensity ridge that connects the bottom two clusters suggests that they can also be served by the intermediate locations. 

\begin{figure}[ht]
	\centering
	\includegraphics[width=\linewidth, trim=8mm 5mm 0cm 2mm, clip]{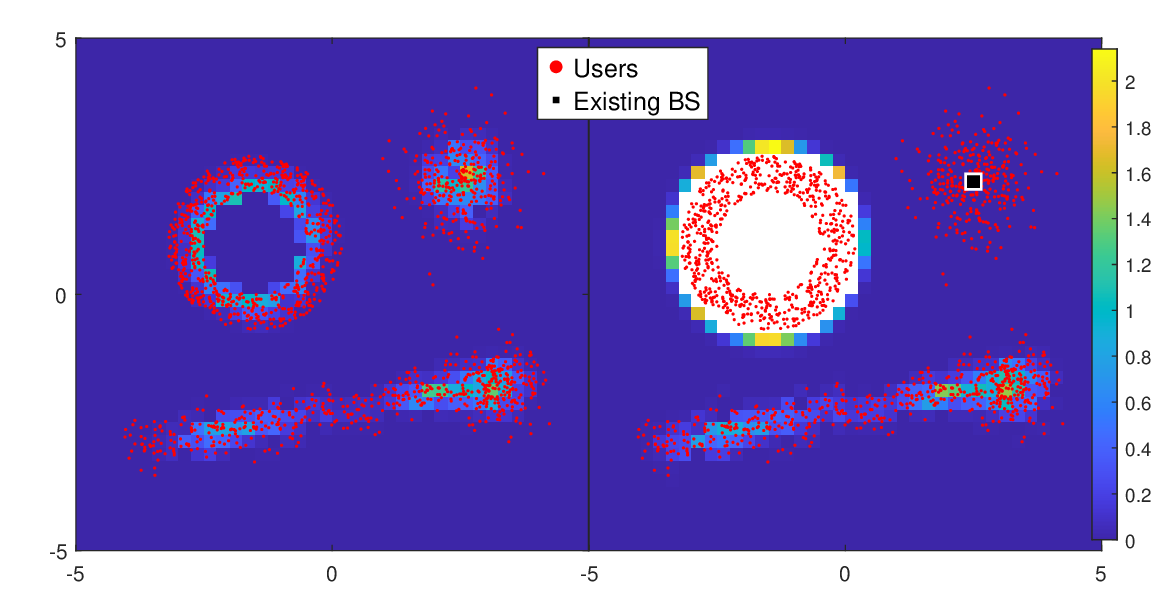}
	\caption{Optimized BS intensity for a scenario with mixed urban demand under two different planning contexts: (a) unconstrained and (b) incremental with an existing BS and an exclusion zone.}
	\label{fig:context}
\end{figure}

To further highlight that the optimized BS intensity represents deployment utility rather than just user density, we consider a mixed urban demand pattern under two different planning contexts. The user distribution is kept identical across the two cases and consists of a combination of localized demand regions, an annular venue, and a corridor. Fig.~\ref{fig:context}(a) considers a greenfield (unconstrained) deployment, whereas Fig.~\ref{fig:context}(b) introduces an incremental deployment over an existing BS as well as a circular exclusion zone. With $N = 5$, the resulting intensity maps differ substantially: the existing BS reduces the value of placing additional mass in the region it serves and the exclusion zone forces the intensity to concentrate along its outer edge. Thus, the optimized intensity captures not only where additional infrastructure is useful, but also how deployment opportunities shift in response to existing infrastructure and spatial constraints.

\subsection{Deployment performance and finite-sample behavior} \label{sec-theory}
We next evaluate the quality of the concrete BS deployments generated from the optimized intensity and compare them with four baseline deployment strategies: two randomized and two deterministic. The \emph{demand-proportional} baseline assigns the BS intensity $\lambda_u$ in proportion to the empirical (training) user count in each feasible cell, while the \emph{uniform} intensity baseline assigns constant BS intensity $\lambda_c$ to each feasible cell. Both these randomized approaches subsequently follow the same validation-based deployment procedure as outlined in Sec. \ref{sec-deploy}. The \emph{greedy} baseline sequentially adds one BS at the feasible candidate location that maximizes training coverage probability, with all previously selected (and existing) BSs contributing to the SINR calculation. Finally, in $K$-means clustering, the training users are divided into $N$ clusters and the BSs are placed at the resulting cluster centroids. In case the cluster centroids lie inside the exclusion zones, they are projected to the nearest point in the feasible region. We note that the optimal deterministic BS deployment can, in theory, be found by maximizing the training coverage probability over the set of feasible BS locations. However, the resulting problem has a non-convex non-smooth objective and non-convex constraints, making it difficult to solve. Further, such deterministic methods return a single deployment and do not directly characterize alternative viable configurations, making them less flexible under operational uncertainty. On the other hand, for randomized methods, alternative high-quality deployments can be obtained from the same optimized intensity without re-solving the problem. 



For these experiments, we consider a synthetic user-location distribution, where $30\%$ of the users are uniformly distributed over an annulus of outer radius 1.7 and thickness 0.7, centered at $(-1,-1)$ while the remaining users follow a Gaussian distribution centered at $(3,3)$. An existing BS is located at the center of the Gaussian component. We consider two cases: scarce deployment with $N = 2$ and an adequate deployment with $N = 4$. 

Recall that the candidate deployments $\{\phi_k\}_{k=1}^K$ generated by the randomized approaches, namely those involving uniform, demand-proportional, and optimized intensities, have $N$ BSs on average. However, a randomly generated deployment $\phi_k$ could contain more or fewer than $N$ BSs, making its coverage probability unsuitable for a direct comparison with the deterministic, namely greedy or $K$-means, methods that place exactly $N$ BSs. For the purpose of comparing concrete deployments under the same physical BS budget, we therefore generate deployments $\{\phi_k\}_{k=1}^K$ that have exactly $N$ BSs. This conditioning is used only for evaluation in Fig. \ref{fig:comparison} and the proposed intensity is still obtained by optimizing the original IPPP objective. Let us denote the test set of user locations by $\cZ_t: = \{\z_m^t\}_{m=1}^{M_t}$ and the empirical coverage probability of a deployment $\phi$ over the user set $\cZ = \{\z_m\}_{m=1}^{\abs{\cZ}}$ by 
\begin{align}
    \hat{p}_c(\phi,\cZ)  = \frac{1}{\abs{\cZ}}\sum\nolimits_{m=1}^{\abs{\cZ}} p_c(\phi,\z_m).
\end{align}
We consider two metrics for evaluating the randomized methods: (a) the conditional complementary cumulative distribution function (CCDF) given by $\frac{1}{K}\sum_{k=1}^K \ind\{\hat{p}_c(\phi_k,\cZ_t) \geq \rho\}$ which gives the fraction of $N$-BS deployments whose coverage probability is at least $\rho$; and (b) the best deployment test coverage probability $\hat{p}_c^{\rm sel} = \hat{p}_c(\hat{\phi}_{K},\cZ_t)$ where the deployment is selected using validation set as in \eqref{phihk}.
A higher CCDF indicates a higher probability of generating a high-quality deployment, while a higher $\hat{p}_c^{\rm sel}$ indicates that the selected deployment is indeed high-quality. 

As can be seen from Fig. \ref{fig:comparison}, the proposed approach yields a markedly more favorable distribution of deployment quality, as compared to the other randomized schemes, for both $N = 2$ and $N = 4$ cases. We also observe that for the $N = 2$ case, the best validation-selected deployments from the demand-proportional, uniform, and optimized intensities have similar coverage probabilities. This is because in the scarce-BS regime, only a few good-quality deployments possible, and sufficiently many candidate draws can yield a good deployment from any starting $\lambda$. However, as evident from the conditional CCDF curves, a large fraction of deployments resulting from the optimized $\lambda$ are high-quality. The separation is clearer in the $N = 4$ case, where the validation-selected deployment of the optimized $\lambda$ has very high coverage. As expected, the demand-proportional intensity often yields poor-quality deployments, as it does not account for the coverage and interference of the installed network. The same is true for the deterministic $K$-means algorithm, which also achieves poor performance in such settings. Interestingly, the greedy placement overcomes both these issues, but since it selects the best BS locations in a sequential manner, the placement of its first BS is often suboptimal. Further, such deployments may overfit a specific scenario, resulting in a lower performance on the test set.

\begin{figure}[ht]
	\centering
	\includegraphics[width=0.9\linewidth]{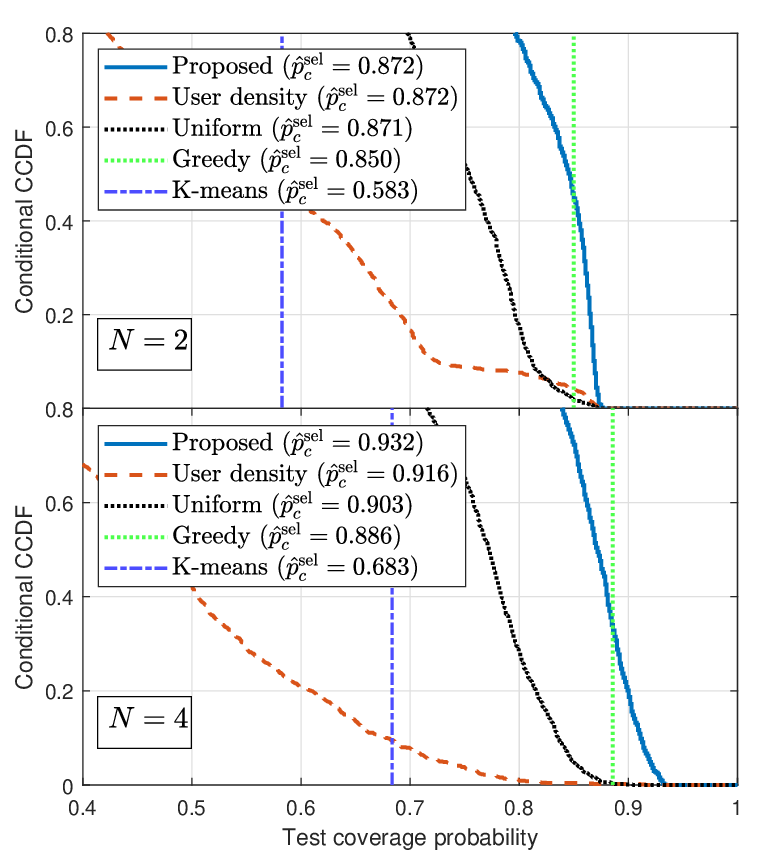}
	\caption{Empirical distribution of $p_c$ evaluated over the common test set, conditioned on the deployments having exactly $N$ BSs. The legend provides test coverage probability of the validation-selected deployment.}
	\label{fig:comparison}
\end{figure}

We next examine empirically the discretization and sampling errors bounded in Sec. \ref{sec-excess} and illustrate that their dependence on $D$ and $M$ follows the trends predicted by the bounds. All subsequent expressions are evaluated for the original mean-$N$ PPP model. Since the optimal intensities appearing in $E_d$ and $E_s$ cannot be computed exactly, we use the following  approximations.
\begin{enumerate}[leftmargin=*]
    \item For each $D$, a solution of \eqref{pe} obtained using a large $M$ approximates $\lamt_\d$. 
    \item Solution of \eqref{pe} obtained using large $M$ and $D$ approximates $\lambda^\star$.
    \item For a given $\lambda$, $p_c(\lambda)$ is estimated by averaging $p_c(\lambda,\z)$ over a large independent test set.
    \item For each $M$, the finite-sample experiment is repeated using independent training sets, so that the average of the resulting errors approximates $E_s$.
    \item Finally, in all cases, the output of the mirror-ascent algorithm is used to approximate the corresponding optimizer.
\end{enumerate}

Fig. \ref{fig:theoryDM} shows the approximate discretization and sampling errors as functions of $\sqrt{D}$ and $\sqrt{M/D}$, respectively. As expected, the discretization error decreases with $\sqrt{D}$, since a finer piecewise-constant intensity can better represent a given continuous intensity. Likewise, as expected the sampling error also decreases with $\sqrt{M/D}$, since a finer grid increases the size of the vector being estimated and hence requires more user-location samples. The trends are similar for both values of $N$. For the sampling error, error bars show 95\% confidence intervals over 10 independent user realizations. 

\begin{figure}[t]
	\centering
	\includegraphics[width=\linewidth]{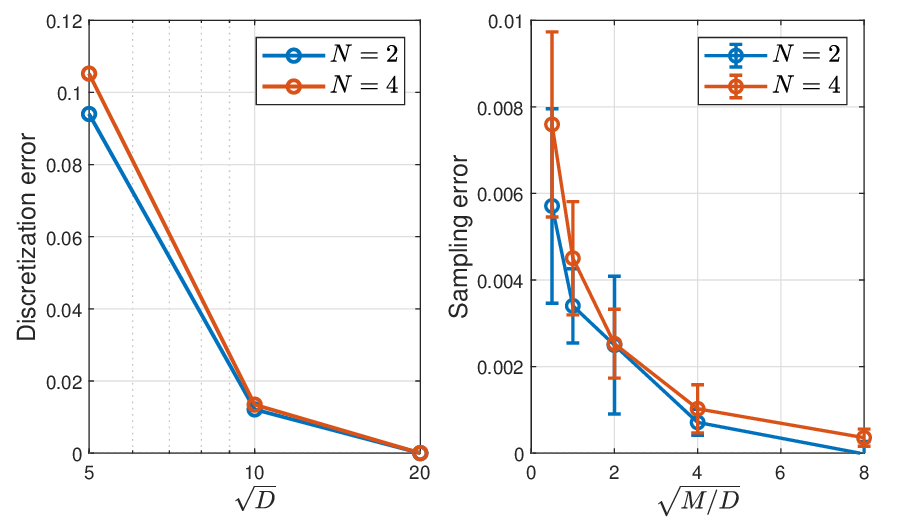}
	\caption{Discretization and sampling errors decrease inversely with $\sqrt{D}$ and $\sqrt{M/D}$, consistent with the theoretical bounds. }
	\label{fig:theoryDM}
\end{figure}

We also examine the coverage probability of the validation-selected deployment as a function of $K$ and $M_v$. Fig. \ref{fig:theoryKM} shows the coverage probability for the selected deployment, estimated by averaging over a large number of independent test user-location samples. For comparison, we also show the coverage probability obtained by oracle selection using the test set, as well as the average coverage probability $p_c(\blamt)$ of the optimized point process. We observe that as $K$ increases, the coverage improves rapidly at first and then more gradually, reflecting the benefit of generating a larger candidate set. Further, the gain from increasing $K$ is larger for $M_v = 200$ than for $M_v = 10$, since a larger validation set more reliably identifies the better deployments. The curves corresponding to different values of $M_v$ are close, suggesting that modest values of $K$ and $M_v$ may be sufficient to obtain a deployment that is better than $p_c(\bar\lambda)$ with high probability. The worst-case dependence on $K$ and $M_v$ obtained in Lemma \ref{eplem} can therefore be conservative in structured deployment problems. 

\begin{figure}[t]
	\centering
	\includegraphics[width=0.9\linewidth]{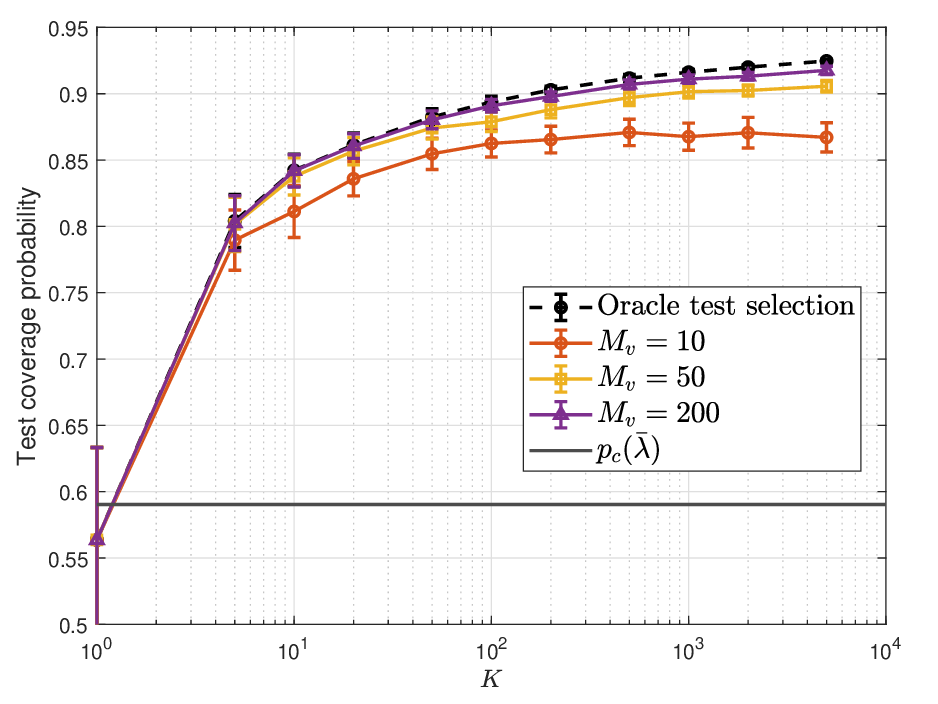}
	\caption{Test coverage probability of selected deployment for $N = 4$.}
	\label{fig:theoryKM}
\end{figure}

\subsection{Incremental deployment in Chicago}\label{sec:numchicago}
Finally, we consider incremental deployment over Chicago. For the user distribution, we utilize Facebook high resolution population density maps, which estimate resident population on approximately 30 m grid cells from satellite imagery and census data \cite{Facebook}. The resulting population map is used as a stand-in for the user-location distribution and deployment region, with cells outside its convex hull or predominantly within large water, wetland, or protected open-space regions treated as exclusion zones. Further, for existing infrastructure, we utilize AT\&T cell records obtained from OpenCellID \cite{opencellid}. Nearby records belonging to the same carrier are merged into spatial site proxies. After geographic and clustering filters, the experiment contains $56$ existing AT\&T site proxies, which are treated as existing BSs. The 14,646 user-location records are split into 70\% training, 15\% validation, and 15\% test samples.

The existing deployment is highly nonuniform and contains large spatial regions where a limited incremental deployment can be useful. Fig.~\ref{fig:chicago} shows the user demand, existing site proxies, and the optimized BS intensity for mean $N = 5$. We observe that the optimized intensity is concentrated in regions having both substantial demand and comparatively sparse existing infrastructure. The experiment demonstrates that the proposed framework can be used to obtain insights regarding incremental deployment in existing urban centers. 

\begin{figure}[t]
	\centering
	\includegraphics[width=\linewidth, trim=45mm 11mm 45mm 5mm, clip]{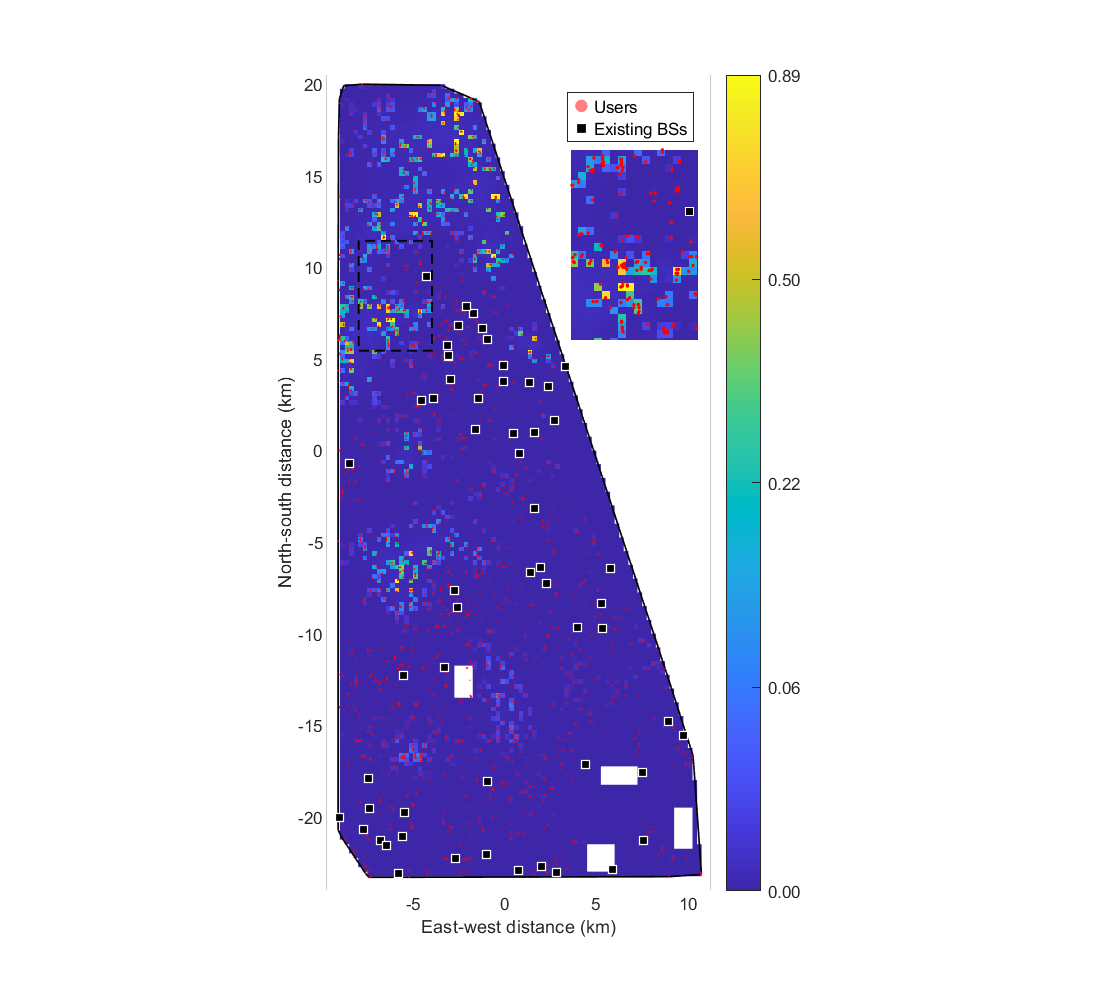}
	\caption{Optimized incremental BS intensity in Chicago with existing infrastructure and exclusion regions. The inset shows the zoomed-in version of the dotted rectangle.}
	\label{fig:chicago}
\end{figure}

\section{Conclusion}
We studied base station (BS) deployment optimization using stochastic geometric principles, where the deployment intensity is optimized using sampled user-location data. The proposed framework converts the infinite-dimensional problem into a finite-dimensional empirical problem that can be solved using standard optimization algorithms. We established bounds on the discretization and sampling errors arising from these approximations, as well as the error incurred in selecting a concrete BS deployment from the optimized intensity. The framework was applied to maximize coverage probability in downlink cellular networks under different association rules and practical deployment constraints. Numerical experiments demonstrate how the optimized intensity captures the effects of user distribution, existing infrastructure, and deployment constraints. More importantly, unlike a single deterministic deployment, the resulting intensity reveals broader spatial deployment priorities and alternative viable deployment regions, providing useful insights for cellular network planning.

\appendices
\section{Bound in Remark \ref{ed2}} \label{ed2proof}
The bound follows directly from the application of a well-known Poincaré–Wirtinger inequality \cite[Sec. 5.8.1]{muscat2024functional}, which yields:
\begin{align}
	\int_{\cA_k} \norm{\mu^\star(\x)-\mu_{\d}^\star(\x)} d\x \leq \frac{\text{diam}(\cA_k)}{2}\int_{\cA_k}\abs{\nabla \mu^\star(\x)} d\x.
\end{align}
Multiplying both sides by $N$, summing over $k = 1, \ldots, D$, and using the bound $\text{diam}(\cA_k) \leq \sqrt{\tfrac{2A}{D}}$ for square-sub-regions, we obtain
\begin{align}
	\sum_{k=1}^D \int_{\cA_k}\norm{\lambda^\star(\x) - \lambda^\star_{\d}(\x)}d\x \leq NB_\star\sqrt{\frac{A}{2D}}
\end{align}
Substituting into \eqref{wphi} and proceeding as in the proof of Lemma \ref{edlem}, we obtain
\begin{align}
	E_\d \leq \frac{LNAB_\star}{D}
\end{align}
which is the required result. 

\section{Proof of Lemma \ref{rad}}\label{radproof}
The proof follows along the lines of the symmetrization result in \cite[Sec. 4.5]{bach2024learning}. Let $\{\z_m'\}_{m=1}^M$ be an independent copy of the training sample with the same distribution as $\{\z_m\}_{m=1}^M$. Defining
\begin{align}
	\Rh'(\lam_\d )&:=\frac{1}{M}\sum_{m=1}^M r(\lam_\d,\z_m')
\end{align} 
and using $\mathbb{E}'$ to denote the expectation with respect to $\{\z_m'\}_{m=1}^M$, we can write
\begin{align}
	\max_{\lam_\d \in\Lam}\left(R(\lam_\d ) - \Rh(\lam_\d )\right) &= \max_{\lam_\d\in \Lam}\left(\Exp{\Rh'(\lam_\d )} - \Rh(\lam_\d )\right) \nonumber\\
	&\hspace{-1cm}= \max_{\lam_\d\in \Lam}\left(\Exp{\Rh'(\lam_\d )-\Rh(\lam_\d )}\right) \nonumber\\
	&\hspace{-1cm}\leq \Exp{\max_{\lam_\d\in \Lam}\left(\Rh'(\lam_\d )-\Rh(\lam_\d )\right)} \label{radproof1}
\end{align}
where we have used the convexity of the max function. Next, we use the observation that a symmetric random variable $X$ has the same distribution as $\varepsilon X$, where $\varepsilon$ is an independent Rademacher random variable. Applying it to $X = r(\lam,\z'_m)-r(\lam,\z_m)$, substituting into \eqref{radproof1}, and taking expectation with respect to $\{\z_m\}$, we obtain
\begin{align}
	&\Ex{\Exp{\max_{\lam_\d\in \Lam}\left(\Rh'(\lam_\d )-\Rh(\lam_\d )\right)}} \nonumber\\
	&=	\Ex{\max_{\lam_\d \in\Lam} \frac{1}{M}\sum_{m=1}^M \varepsilon_m\left(r(\lam_\d,\z'_m)-r(\lam_\d,\z_m)\right)} \nonumber\\
	&\leq \Ex{\max_{\lam_\d\in \Lam} \frac{1}{M}\sum_{m=1}^M \varepsilon_mr(\lam_\d,\z'_m)} \nonumber\\
	&+ \Ex{\max_{\lam_\d\in \Lam} \frac{1}{M}\sum_{m=1}^M -\varepsilon_mr(\lam_\d,\z_m)}.\label{radproof2}
\end{align}
where the expectations on the right are with respect to $\{\z_m,\z_m',\varepsilon_m\}$. Again, since $\varepsilon_mr(\lam_\d,\z'_m)$ has the same distribution as $-\varepsilon_mr(\lam_\d,\z_m)$, the last two terms in \eqref{radproof2} are the same. Hence, substituting into \eqref{radproof1}, we obtain 
\begin{align}
	&\Ex{\max_{\lam_\d \in\Lam}\left(R(\lam_\d ) - \Rh(\lam_\d )\right)} \nonumber\\
	& \leq 2\Ex{\max_{\lam_\d\in \Lam} \frac{1}{M}\sum_{m=1}^M \varepsilon_mr(\lam_\d,\z'_m)} =2R_M
\end{align}
which is the required result. The proof for the second inequality also follows similarly. 

\section{Bound in \texorpdfstring{\eqref{ep1bound}}{the bound}}\label{app_ep1}
Consider i.i.d. random variables $\{\sX_k\}_{k=1}^K$ in $[0,B_r]$ with cumulative distribution $F(x)$ and mean $\mu:=\EE\sX_k$. For any $0\leq t\leq\mu$, applying Markov's inequality to the nonnegative random variable $B_r-\sX_k$  gives \cite[(1.5)]{bach2024learning}
\begin{align}
F(\mu-t)
&=\Px{\sX_k\leq\mu-t} =\Px{B_r-\sX_k\geq B_r-\mu+t}
\nonumber\\
&\leq \tfrac{B_r-\mu}{B_r-\mu+t} \leq 1-\tfrac{t}{B_r}.
\end{align}
Since $\sX_{\max}:=\max_k\sX_k$ has distribution $F(t)^K$, the tail-integral formula yields:
\begin{align}
&\EE[(\mu-\sX_{\max})_+] =\int_0^\mu \Px{\sX_{\max}<\mu-t} dt\\
&=\int_0^\mu F(\mu-t)^K\,dt \leq \int_0^{B_r}\left(1-\frac{t}{B_r}\right)^Kdt
\leq \tfrac{B_r}{K}.
\end{align}

\section{Lipschitz Continuity of coverage probabilities in \texorpdfstring{\eqref{insteq}}{instantaneous case} and \texorpdfstring{\eqref{avSINR1}}{average case}}\label{lips-proof}
In this section we show that the coverage probabilities in \eqref{insteq} and \eqref{avSINR1} are $L_\ell(e^{-1}+4^{-1})$-Lipschitz and $L_\ell e^{-1}$-Lipschitz respectively, with respect to the $W_1$ metric, hence verifying Assumption \ref{Wlip}. The approach is similar for both the coverage probability expressions. The functional mean-value theorem gives \cite{sahoo1998mean}
\begin{align}
	&r(\lambda,\z)-r(\lambda',\z) = \int_0^1 \int_{\cA} \varphi_t(\y) (\lambda(\y)-\lambda'(\y)) d\y dt \nonumber\\
	&= \int_{\cA} \varphi(\y) (\lambda(\y)-\lambda'(\y)) d\y 
\end{align}
where $\varphi_t(\y) := \frac{dr(\lambda_t,\z)}{d\lambda}(\y)$, $\lambda_t = (1-t)\lambda'+t\lambda$, and $\varphi(\y):=\int_0^1\varphi_t(\y)dt$. If $\varphi$ is $L_\varphi$-Lipschitz, we can write $\frac{1}{L_\varphi}\varphi \in \cL_1$ and use \eqref{Wdef} to obtain
\begin{align}
	\abs{r(\lambda,\z)-r(\lambda',\z)}&\leq L_\varphi W_1(\lambda,\lambda')
\end{align}
It remains to determine $L_\varphi$ for the two cases. 

\subsection{Lipschitz continuity of \texorpdfstring{$p_c(\cdot,\z)$}{coverage probability} in \texorpdfstring{\eqref{insteq}}{instantaneous case}} We begin by calculating its functional derivative \cite[Sec. 2.2]{cassel2013variational}:
\begin{align}
	\frac{dp_c(\lambda,\z)}{d\lambda}(\y) =  h(\lambda,\y) - \int_{\cA}\lambda(\x)h(\lambda,\x)b(\x,\y)d\x. \label{pcdiff}
\end{align}
We observe from the triangle inequality that 
\begin{align}
	&\norm{\int_0^1\hspace{-2mm}\nabla_\y \varphi_t(\y)dt} \leq \int_0^1 \hspace{-2mm}\norm{\nabla_\y \varphi_t(\y)}dt \leq \sup_{t\in[0,1]}\norm{\nabla_\y \varphi_t(\y)} \nonumber\\
	&\leq \sup_{t\in[0,1]}\norm{\nabla_\y h(\lambda_t,\y)} \nonumber\\
	&+  \sup_{t\in[0,1]}\norm{\int_{\cA}\lambda_t(\x)h(\lambda_t,\x)\nabla_\y b(\x,\y)d\x} =:T_1+T_2\label{t1t2}
\end{align}
For the first term, for any $t \in [0,1]$, the triangle inequality yields
\begin{align}
	&\norm{\nabla_\y h(\lambda_t,\y)} = h(\lambda_t,\y)\norm{\nabla_{\y} a(\y) + \int_{\cA}\lambda_t(\u) \nabla_{\y} b(\y,\u)d\u} \nonumber\\
	&\leq h(\lambda_t,\y)\left(\norm{\nabla_{\y} a(\y)} + \norm{\int_{\cA}\lambda_t(\u) \nabla_{\y} b(\y,\u)d\u}\right).\nonumber
\end{align}
Recall that since $e^{L_\ell r}\ell(r)$ is increasing in $r$, we have that 
\begin{align}
	\sup_{r \geq 0}\frac{\abs{\ell'(r)}}{\ell(r)} \leq L_\ell. \label{Lb} 
\end{align}
For $\y \neq \z$, since $\norm{\nabla_{\y}d_{\y}} = 1$, we can write
\begin{align}
	\norm{\nabla_\y a(\y)} &= \frac{\kappa}{\ell^2(d_\y)}\abs{\ell'(d_\y)} \leq L_\ell\frac{\kappa}{\ell(d_\y)} = L_\ell a(\y) \\
	\nabla_\y b(\y,\u) &= -\frac{\tau \ell(d_\u)}{(\tau \ell(d_\u) + \ell(d_\y))^2}\ell'(d_\y)(\nabla_\y d_\y)
\end{align}
so using the fact that $\ell(d_\u) \geq 0$, we obtain
\begin{align}
	\norm{\int_{\cA}\lambda_t(\u)\nabla_\y b(\y,\u)d\u} \leq L_\ell \int_{\cA} \lambda_t(\u) b(\y,\u)d\u
\end{align}

Combining these and substituting $w:=a(\y) + \int_{\cA}\lambda_t(\u)b(\y,\u)d\u$, we obtain 
$\norm{\nabla_\y h(\lambda_t,\y)}
\leq L_\ell w e^{-w} \leq \frac{L_\ell}{e}$. Hence the bound on $T_1$ becomes $T_1 \leq  \frac{L_\ell}{e}$. 
For the second term in \eqref{t1t2}, observe that 
\begin{align}
	\nabla_\y b(\x,\y) = \frac{\tau \ell(d_\x)\ell'(d_\y)}{(\tau \ell(d_\y) + \ell(d_\x))^2}(\nabla_\y d_\y)
\end{align}
so using $\norm{\nabla_{\y}d_{\y}} = 1$ for $\y \neq \z$  and the inequality $4ab \leq (a+b)^2$, we have the bound
\begin{align}
	&\norm{\int_{\cA}\lambda_t(\x)h(\lambda_t,\x)
\nabla_\y b(\x,\y)d\x} \nonumber\\
&\leq \frac{\abs{\ell'(d_\y)}}{\ell(d_\y)}\int_{\cA}\lambda_t(\x)h(\lambda_t,\x)\frac{\tau \ell(d_\x)\ell(d_\y)}{(\tau \ell(d_\y) + \ell(d_\x))^2}d\x \\
	&\leq \frac{L_\ell}{4} \int_{\cA}\lambda_t(\x)h(\lambda_t,\x)d\x = \tfrac{1}{4}L_\ell p_c(\lambda_t,\z) \leq L_\ell/4
\end{align}
for $t\in [0,1]$. Therefore we have that $T_2 \leq L_\ell/4$. Adding the bounds on $T_1$ and $T_2$, we can take $L_\varphi =L_\ell(e^{-1}+4^{-1}) < 2L_\ell/3$ and hence $p_c(\lambda,\z)$ is $\tfrac{2}{3}L_\ell$-Lipschitz in $\lambda$ with respect to the $W_1$ metric.

\subsection{Lipschitz continuity of \texorpdfstring{$p_c^a(\cdot,\z)$}{coverage probability} in \texorpdfstring{\eqref{avSINR1}}{average case}}
Defining 
\begin{align}
    p(r,\z):=\int_{\cB(\z,r)}\lambda(\y)\d\y = \int_{\cA}\lambda(\y)\ind_{\{d_\y \leq r\}}\d\y \label{prz}
\end{align}
the coverage probability in \eqref{avSINR1} takes the form
\begin{align}
    p_c^a(\lambda,\z) = \int_0^\infty h^a(\lambda,r) e^{-p(r,\z)}\frac{dp(r,\z)}{\d r}\d r
\end{align}
where $d_\y:=\norm{\y-\z}$ and the expression for $h^a(\lambda,r)$ is provided in \eqref{aveq}. Note that for a fixed $\y\in \cA$, 
\begin{align}
    \frac{\d p(r,\z)}{\d\lambda}(\y) &= \ind_{\{d_\y\leq r\}} & \frac{\d}{\d\lambda(\y)}\left[\frac{\d p(r,\z)}{\d r}\right] &=\delta(r-d_\y),
\end{align}
where $\delta(\cdot)$ denotes the Dirac delta function. Hence, the functional derivative of $p_c^a$ is given by 
\begin{align}
&\frac{\d p_c^a(\lambda,\z)}{\d\lambda}(\y)
= -\int_0^\infty h^a(\lambda,r)e^{-p(r,\z)} \frac{\ind_{\{d_\y>r\}}}
{1+\frac{\ell(r)}{\tau\ell(d_\y)}} \frac{\d p(r,\z)}{\d r}\d r
\nonumber\\
&\hspace{1cm}- \int_0^\infty h^a(\lambda,r)e^{-p(r,\z)} \ind_{\{d_\y\leq r\}} \frac{\d p(r,\z)}{\d r}\d r \nonumber\\
&\hspace{1cm}+
\int_0^\infty h^a(\lambda,r)e^{-p(r,\z)} \delta(r-d_\y)\d r \\
&=-\int_0^{d_\y} \frac{h^a(\lambda,r)}{1+\frac{\ell(r)} {\tau\ell(d_\y)}} f_{R_\z}(r)\d r - \int_{d_\y}^{\infty} h^a(\lambda,r)f_{R_\z}(r)\d r\nonumber\\ 
&\hspace{1cm}+ h^a(\lambda,d_\y)e^{-p(d_\y,\z)} =: \varphi^a(\y)
\label{eq:aa_functional_derivative_final}
\end{align}
Since all the terms in $\varphi^a(\y)$ depend on
$\y$ only through $d_\y$, let $s:=d_\y$, so that it suffices to use the fact that $\abs{d_\y - d_{\y'}}\leq \norm{\y-\y'}$ and only bound  $\abs{\frac{d}{ds}\varphi_t^a(s)}$ where $\varphi_t^a(s)$ is defined in \eqref{eq:aa_functional_derivative_final} with $\lambda$ replaced by
$\lambda_t =(1-t)\lambda'+t\lambda$. With a slight abuse of
notation, define
\begin{align}
a(s) &:= \frac{\kappa}{\ell(s)} & b(r,s)&:=\frac{\tau\ell(s)}{\tau\ell(s)+\ell(r)},
\end{align}
which are simply \eqref{abdef} when the serving and interfering BSs are
at distances $r$ and $s$, respectively. Then
\begin{align}
\varphi_t^a(s)
={}&h^a(\lambda_t,s)e^{-p_t(s,\z)}
-\int_0^s h^a(\lambda_t,r)b(r,s)f_{R_\z,t}(r)\d r
\nonumber\\
&-\int_s^\infty h^a(\lambda_t,r)f_{R_\z,t}(r)\d r
\label{eq:aa_varphi_s}
\end{align}
where $p_t$ is defined as in \eqref{prz} and $f_{R_\z,t}$ in \eqref{eq:fr_inhom} with $\lambda$ replaced with $\lambda_t$. Hence, we obtain 
\begin{align}
&h^a(\lambda_t,s)e^{-p_t(s,\z)} = \nonumber\\
&\exp\left(-a(s)-p_t(s,\z)-\int_{\cA\setminus\cB(\z,s)}\hspace{-1cm}\lambda_t(\u)b(s,d_\u)\d\u\right).
\label{eq:aa_ht}
\end{align}
Using the Leibniz rule and the fact that $b(s,s)=\tau/(1+\tau)$, we obtain 
\begin{align}\label{eq:aa_ht_derivative}
&\frac{\d}{\d s} \left[h^a(\lambda_t,s)e^{-p_t(s,\z)}\right] = -h^a(\lambda_t,s)e^{-p_t(s,\z)}  \times \\ 
&\left[ a'(s) +\int_{\cA\setminus\cB(\z,s)} \lambda_t(\u) \frac{\partial b(s,d_\u)}{\partial s}\d\u +\frac{1}{1+\tau}\frac{\d p_t(s,\z)}{\d s} \right].\nonumber
\end{align}
Likewise, differentiating the two integrals in
\eqref{eq:aa_varphi_s} gives
\begin{align}
&\frac{\d}{\d s} \left[ \int_0^s h^a(\lambda_t,r)b(r,s)f_{R_\z,t}(r)\d r \right]
\nonumber\\
&\hspace{2cm}= \frac{\tau}{1+\tau} h^a(\lambda_t,s)e^{-p_t(s,\z)} \frac{\d p_t(s,\z)}{\d s} \nonumber\\
&\hspace{2cm}+ \int_0^s h^a(\lambda_t,r) \frac{\partial b(r,s)}{\partial s} f_{R_\z,t}(r)\d r
\label{eq:aa_int1_derivative}\\
&\frac{\d}{\d s} \left[ \int_s^\infty h^a(\lambda_t,r)f_{R_\z,t}(r)\d r \right] 
\nonumber\\
&\hspace{2cm}= -h^a(\lambda_t,s)e^{-p_t(s,\z)} \frac{\d p_t(s,\z)}{\d s}.
\label{eq:aa_int2_derivative}
\end{align}
Substituting \eqref{eq:aa_ht_derivative}--\eqref{eq:aa_int2_derivative}
into the derivative of \eqref{eq:aa_varphi_s}, the terms containing
$\frac{\d p_t(s,\z)}{\d s}$ cancel, and we obtain
\begin{align}
(\varphi_t^a)'(s)=& -h^a(\lambda_t,s)e^{-p_t(s,\z)} \left[
a'(s) +\int_{\cA\setminus\cB(\z,s)} \hspace{-1cm}\lambda_t(\u) \frac{\partial b(s,d_\u)}{\partial s}\d\u \right] \nonumber\\
&- \int_0^s h^a(\lambda_t,r) \frac{\partial b(r,s)}{\partial s}
f_{R_\z,t}(r)\d r. \label{eq:aa_varphi_derivative}
\end{align}
whose terms can be bounded using the same arguments as those used for $T_1$ and $T_2$ in the instantaneous-association proof. Specifically, defining
\begin{align}
w:=a(s)+ \int_{\cA\setminus\cB(\z,s)} \lambda_t(\u)b(s,d_\u)\d\u,
\end{align} 
we have $h^a(\lambda_t,s)e^{-p_t(s,\z)}\leq e^{-w}$, and therefore
the magnitude of the first term in \eqref{eq:aa_varphi_derivative}
is bounded as in $T_1$ by $L_\ell we^{-w}\leq \frac{L_\ell}{e}$.  For the second term, we can likewise bound $-\frac{\partial b(r,s)}{\partial s} \leq \frac{L_\ell}{4}$ so that
\begin{align}
-\int_0^s h^a(\lambda_t,r) \frac{\partial b(r,s)}{\partial s} f_{R_\z,t}(r)\d r \leq \frac{L_\ell}{4}p_c^a(\lambda_t,\z) \leq \frac{L_\ell}{4}. \nonumber
\end{align}
Different from the earlier proof, the first term in \eqref{eq:aa_varphi_derivative} is nonpositive, while the second is nonnegative, and hence $-\frac{L_\ell}{e} \leq\varphi_t'(s)\leq\frac{L_\ell}{4}$ which implies that we can take $L_\varphi = \frac{L_\ell}{e}$. In conclusion, $p_c^a(\lambda,\z)$ is $L_\ell/e$-Lipschitz with respect to the $W_1$ metric. 

\section{Incremental Deployment}\label{incremental}
As earlier, let $\Phi_0$ denote the set of existing BSs. Under the instantaneous-SINR association rule, a user $\z$ may be covered either through one of the installed BSs $\u \in \Phi_0$ with probability 
\begin{align}
    p_c^1 = \sum_{\u\in\Phi_0} h(\lambda,\u)\prod_{\v \in \Phi_0 \setminus \{\u\}}(1-b(\u,\v))\label{pc1-incre}
\end{align}
or through one of the newly deployed BSs with probability
\begin{align}
    p_c^2 = \int_{\cA} \lambda(\x)h(\lambda,\x)\prod_{\u\in\Phi_0}(1-b(\x,\u))d\x.\label{pc2-incre}
\end{align}
Since $\tau > 1$, at most one BS can satisfy the SINR threshold. Therefore, the total coverage probability is given by $p_c^1+p_c^2$. Coverage probability of a specified set $\Phi_0$ of BSs can likewise be obtained by substituting $\lambda \equiv 0$ in \eqref{pc1-incre}. 

Under average-SINR association, the serving BS is the nearest BS to the user. Let $\u_\z$ denote the existing BS that is closest to $\z$ and $r_\u := \norm{\u_\z - \z}$. Recall that $\x_\z$ is the newly deployed BS closest to $\z$ and $R_\z = \norm{\x_\z - \z}$. There are two distinct possibilities: $R_\z \geq r_\u$ so that the user is served by an existing BS, and $R_\z < r_\u$ so that the user is served by a new BS. The contribution to the coverage probability from the first case is
\begin{align}
    p_c^1 &= \prod_{\v \in \Phi_0 \setminus \{\u_\z\}} (1-b(\u_\z,\v))\\
    &\times \exp\left(-a(\u_\z) -\int_{\cB(\z,r_\u)}\hspace{-6mm}\lambda(\y)\d\y -\int_{\cA\setminus\cB(\z,r_\u)}\hspace{-8mm}\lambda(\y)b(\u_\z,\y)\d\y\right)\nonumber
\end{align}
where $a(\cdot)$ is defined in \eqref{abdef}. When a newly deployed BS serves $\z$, the coverage probability is given by 
\begin{align}
    p_c^2 &= \int_{0}^{r_\u} \exp\left(-\frac{\tau\sigma^2}{P\ell(r)}-\int_{\cA\setminus\cB(\z,r)}\frac{\lambda(\y)\d\y}{1+\frac{\ell(r)}{\tau\ell(\norm{\y-\z})}}\right) \nonumber\\
&\qquad\times \prod_{\v\in\Phi_0} \frac{\ell(r)}{\ell(r)+\tau\ell(\norm{\v-\z})} f_{R_{\z}}(r)\d r.
\end{align}
Finally, the total coverage probability is given by $p_c^1+p_c^2$. 

\section{Convexity in the Noise-Limited Case}
\label{sec:noise_limited}
In this section, we show that $p_c(\lambda)$ for IPPP deployment is concave when the network is noise-limited, i.e., when the aggregate interference power $I(\z)$ is much smaller than the noise power $\sigma^2$. 

\begin{lemma}\label{lem:noise_limited}
In the noise-limited case, the coverage probability expressions in \eqref{insteq} and \eqref{avSINR1} are both concave in $\lambda$. 
\end{lemma}
\begin{IEEEproof}
In the instantaneous-SINR association, the noise-limited coverage probability is given by \cite[Sec.~II]{andrews2016primer}:
\begin{align}
p_c(\lambda,\z) = 1-\exp\left(-\int_{\cA}\lambda(\x)e^{-a(\x)}\d\x\right)
\label{eq:noise_ia}
\end{align}
which is concave since the function $1-e^{-x}$ is concave in $x$ and the integral is linear in $\lambda$. 

In the average-SINR association, the noise-limited version of \eqref{avSINR1} becomes 
\begin{align}
p_c^a(\lambda,\z) &= \int_0^\infty \exp\left(-\frac{\tau\sigma^2}{P\ell(r)}\right)f_{R_{\z}}(r)\d r
\label{eq:noise_aa} \\
&= \int_0^\infty\Px{R_{\z}\leq r} \d\left[-\exp\left(-\frac{\tau\sigma^2}{P\ell(r)}\right)\right]
\label{eq:noise_aa_concave}
\end{align}
where $\Px{R_{\z}\leq r} = 1-\exp\left(-\int_{\cB(\z,r)}\lambda(\y)\d\y\right)$ is concave in $\lambda$ and the measure $\d\left[-\exp\left(-\frac{\tau\sigma^2}{P\ell(r)}\right)\right]$  is nonnegative. Hence,
$p_c^a(\lambda,\z)$ is a nonnegative weighted integral of concave functions and is therefore concave in $\lambda$.

Finally,  expectation over $\z$ and empirical averaging preserve
concavity, and hence the coverage probability is concave in both the cases. 
\end{IEEEproof}

In the noise-limited case therefore, mirror ascent converges to the global optimum and $E_o$ can be made arbitrarily small. To obtain the bounds in closed-form however, we must first establish that the discretized gradients of $p_c(\lam_\d,\z)$ and $\p_c^a(\lam_\d,\z)$ have bounded components. 

\begin{lemma}
    In the noise-limited case, it holds that 
    \begin{align}
        \norm{\nabla p_c(\lam_\d,\z)}_\infty &\leq \tfrac{A}{D} & \norm{\nabla p_c^a(\lam_\d,\z)}_\infty &\leq \tfrac{A}{D}
    \end{align}
    where the expressions for $p_c$ and $p_c^a$ are provided in \eqref{insteq} and \eqref{avSINR1}, respectively. 
\end{lemma}
\begin{IEEEproof}
    Observe that in \eqref{pcdiff}, both $h(\lambda,\y)$ and $b(\x,\y)$ lie in the interval $[0,1]$. Therefore, the right-hand side is bounded as $\abs{\frac{dp_c(\lambda,\z)}{d\lambda}(\y)}\leq 1$. The required bound follows as
    \begin{align}
        \norm{\nabla p_c(\lam_\d,\z)}_\infty &= \max_i \abs{[\nabla p_c(\lam_\d,\z)]_i} \\
        &= \max_i \abs{\int_{\cA_i} \frac{dp_c(\lambda,\z)}{d\lambda}(\y)\d\y} \leq \abs{\cA_i} = \tfrac{A}{D}.\nonumber
    \end{align}

    For the average-SINR association case,  we can again use these bounds and the fact that $0\leq \left(1+\frac{\ell(r)}{\tau\ell(d_\y)}\right)^{-1}\leq 1$ in \eqref{eq:aa_functional_derivative_final}, so as to conclude that $\abs{\frac{\d p_c^a(\lambda,\z)}{\d\lambda}(\y)} \leq 1$. Therefore, proceeding as above, we obtain the required bound. 
    \end{IEEEproof}

Since the stochastic gradient $\g^{(t)}$ in \eqref{sgdef} is unbiased and, from the preceding lemma, satisfies $\norm{\g^{(t)}}_\infty\leq A/D$, we can directly apply the standard stochastic mirror-ascent bound for concave objectives \cite{nemirovski2009robust}. In particular, using $\bar{\lam}_\d^T$ as the step-size-weighted average of the iterates and starting from uniform initialization and constant step-size, we have
\begin{align}
E_o \leq N\sqrt{\frac{2\log D}{T}}.
\end{align}
We can then add this bound to the other terms in Thm. \ref{thm:e2e} to obtain the full end-to-end bound. 

\section{Effect of noise power}\label{noise_sweep}
In this section, we again consider the experiment with users distributed along an annulus, as in Fig. \ref{fig:basicmaps}(a), but with different values of noise power $\sigma^2$. As expected, Fig. \ref{fig:ring_sweep} shows that when the noise is high, the optimized intensity remains close to the users to reduce path loss. As noise decreases, a more centrally concentrated intensity can provide coverage to users around the annulus. However, as noise is reduced further, interference becomes relatively more important, and deploying multiple BSs at the center results in interference. Spreading the intensity back towards the annulus increases the likelihood that a user is close to one BS while the others are farther away, thereby improving the SINR and coverage probability.

\begin{figure*}[th]
	\centering
	\includegraphics[width=\linewidth, trim = 6mm 2mm 0mm 2mm, clip] {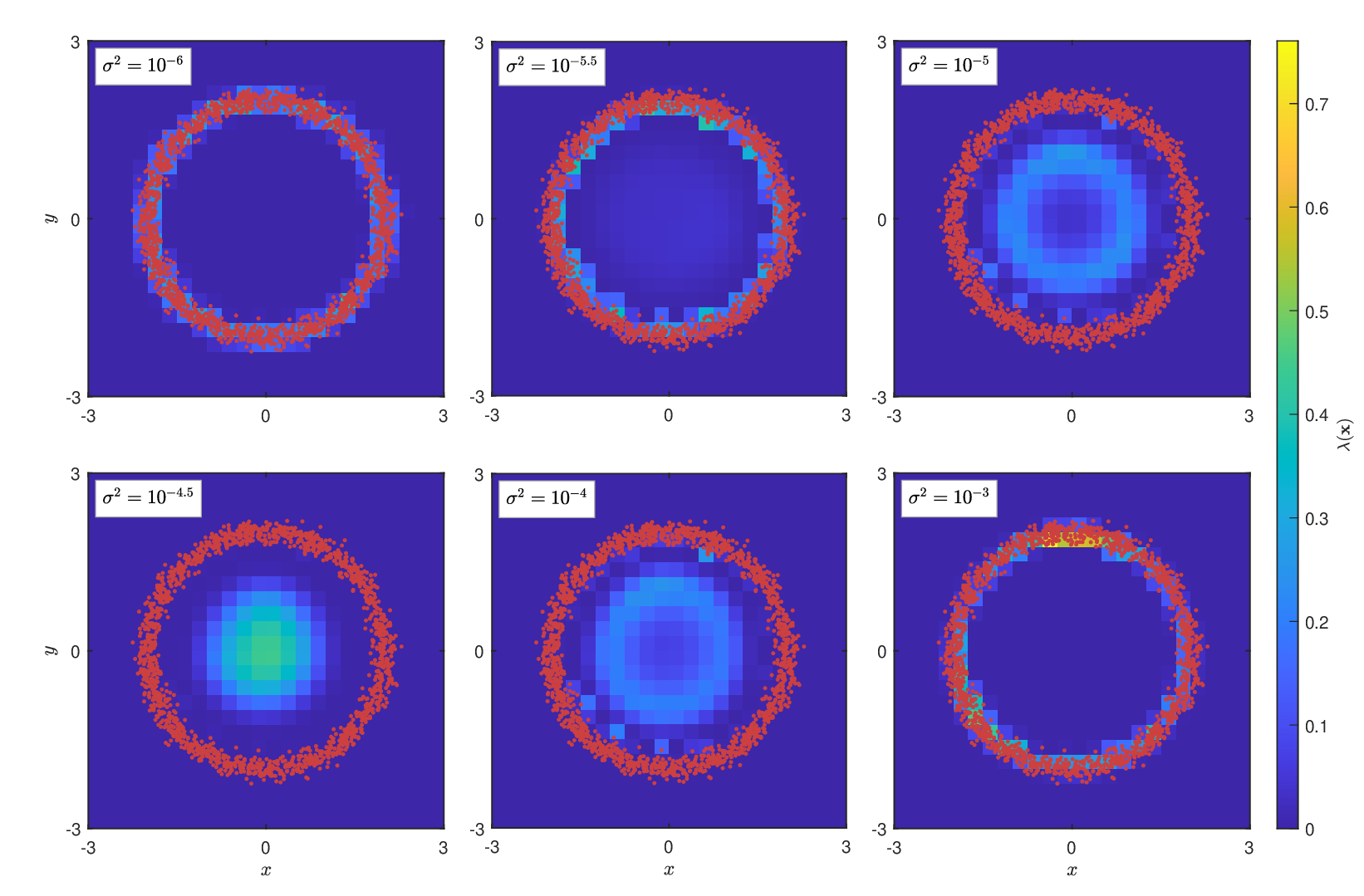}
	\caption{Intensity maps across a range of noise powers. }
	\label{fig:ring_sweep}
\end{figure*}

\section{Effect of the Association Rule}
\label{sec:association_comparison}
In this section, we illustrate how the assumed association rule can affect the optimized intensity. We consider the region $\cA=[-5,5]^2$ with two existing BSs placed symmetrically on both sides of a corridor containing 80\% of the 300 users. We also have a separate group of users at $(3,3)$ containing the remaining users. We take $\sigma^2 = 2\times 10^{-5}$, while the remaining propagation parameters  are the same as specified in Table~\ref{tab:numparams}. The noise parameter is chosen to make the effect of the association rule more visible.

\begin{figure*}[th]
	\centering
	\includegraphics[width=\linewidth, trim = 6mm 2mm 0mm 2mm, clip] {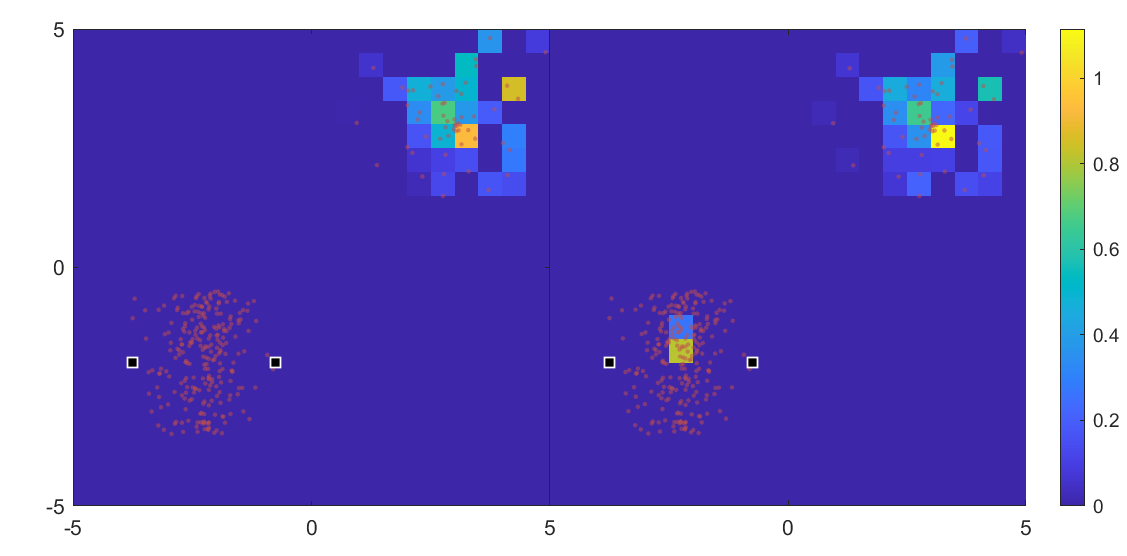}
	\caption{Optimized intensities for the instantaneous-SINR (left) and average-SINR (right) association rules. }
	\label{fig:ia_aa_compare}
\end{figure*}

Fig.~\ref{fig:ia_aa_compare} compares the intensities optimized under instantaneous-SINR association (IA) and average-SINR association (AA). The difference is concentrated in the corridor, where the users lie between the two existing BSs. Under AA, each user is associated with the nearest BS, while the other BS contributes only interference; hence, some additional intensity is allocated within the corridor. Under IA, either existing BS may serve a user depending on the instantaneous fading, reducing the need for additional deployment in this region. The user group centered at $(3,3)$ is far from both existing BSs, and therefore the two association rules yield nearly identical intensities there. Finally, the optimized training coverage probability is $p_c=0.71$ under instantaneous-SINR association and $p_c=0.61$ under average-SINR association. Thus, the association rule can significantly affect both coverage and deployment priorities in regions containing multiple existing BSs.

\footnotesize
 \bibliographystyle{IEEEtran} 
 \bibliography{IEEEabrv,references}   
\end{document}